\documentclass[11pt]{article}
\usepackage{cite}
\usepackage{amsmath,amssymb,amsfonts,amsthm}
\usepackage{graphicx,subfig,float,epstopdf}
\usepackage{algorithmic}
\usepackage{lineno}
\usepackage{latexsym}

\usepackage{graphicx,subfig,float,epstopdf}
\usepackage{color}
\usepackage[colorlinks, linkcolor=red, anchorcolor=green, citecolor=blue]{hyperref}
\usepackage{appendix}
\usepackage{slashbox}
\usepackage{booktabs,multirow}
\usepackage{textcomp}
\usepackage{bm}
\usepackage{color}
\allowdisplaybreaks

\numberwithin{equation}{section}

\newtheorem{theorem}{Theorem}
\newtheorem{lemma}{Lemma}

\newtheorem{remark}{Remark}
\newtheorem{definition}{Definition}
\newtheorem{proposition}{Proposition}

\begin{document}

\title{\Large\bf
Signal and Image Reconstruction with Tight Frames via Unconstrained $\ell_1-\alpha \ell_2$-Analysis Minimizations
\footnotetext{\hspace{-0.35cm}
\endgraf $^\ast$\,Corresponding author.
\endgraf {1. P. Li is with  School of Mathematics and Statistics,  Lanzhou University, Lanzhou 730000, China and also with  Gansu Center of Applied Mathematics, Lanzhou University, Lanzhou 730000, China (E-mail:lp@lzu.edu.cn)}
\endgraf {2. H. ~Ge is with Sports Engineering College, Beijing Sport University,
Beijing 100084, China (E-mail: gehuanmin@163.com)}
\endgraf {3. P. ~Geng is with School of Mathematics and Statistics, Ningxia University,
Yinchuan 750021, China (E-mail: gpb1990@126.com)}
}}
\author{Peng Li$^{1}$, Huanmin Ge$^{2}$ and Pengbo Geng$^{3}$$^\ast$ }
\date{}

\maketitle

\begin{abstract}
In the  paper,
we introduce an unconstrained analysis model based on
 the $\ell_{1}-\alpha \ell_{2}$ $(0< \alpha \leq1)$ minimization
 for the   signal and image reconstruction.
We develop some new technology lemmas for tight frame, and the recovery guarantees
based on the restricted isometry property adapted to frames. 
 The  effective algorithm is established  for the proposed  nonconvex analysis model.
We illustrate the performance of the proposed   model and algorithm  for the  signal  and compressed sensing MRI reconstruction  via
extensive numerical
experiments. And their performance 
is better than that of the existing methods.

\end{abstract}



\section{Introduction}\label{s1}
\noindent

\subsection{Sparse Signal Reconstruction}\label{s1.1}
\noindent

In compressed sensing (CS), a crucial concern  is to reconstruct a high-dimensional
signals from a relatively small number of linear measurements
\begin{equation}\label{systemequationsnoise}
\bm{ b}=\bm{ A}\bm{x}+\bm{e},
\end{equation}
where $\bm{ b}\in\mathbb{R}^m$ is a vector of measurements, $\bm{ A}\in\mathbb{R}^{m\times n}~(m\ll n)$ is a sensing matrix  modeling  the linear measurement process,
$\bm{x}\in \mathbb{R}^{n}$ is an unknown sparse or compressible signal  and $\bm{e}\in\mathbb{R}^{m}$ is a vector of measurement errors, see, e.g., \cite{candes2006stable,donoho2006compressed}.
To reconstruct  $\bm{ x}$, the most intuitive approach is to find the sparsest signal, that is, one solves
via  the $\ell_{0}$ minimization problem:
\begin{equation}\label{VectorL0}
\min_{\bm{x}\in\mathbb{R}^n}\|\bm{x}\|_0~~\text{subject~ to}~~\bm{b}-\bm{ A}\bm{ x}\in\mathcal{B},
\end{equation}
where $\|\bm{x}\|_0$
is  the number of nonzero coordinates of $\bm{ x}$ and $\mathcal{B}$ is a bounded set determined by the error structure. 
It is known that the problem \eqref{VectorL0} is NP-hard for high dimensional signals and
faces challenges in both theoretical and computational (see, e.g., \cite{Donoho2005Stable,2006High,2006On}).
Then  many fast and effective  algorithms  have been developed to
recover $\bm{ x}$  from \eqref{systemequationsnoise}. The Least Absolute Shrinkage and Selector Operator (Lasso) is
 among the most well-known algorithms, i.e.,
\begin{equation}\label{VectorL1-Lasso}
\min_{\bm{ x}\in\mathbb{R}^n}~\lambda\|\bm{ x}\|_1+\frac{1}{2}\|\bm{ Ax} -\bm{b}\|_2^2,
\end{equation}
where $\lambda> 0$ is a parameter to balance the data fidelity term $\|\bm{Ax} -\bm{ b}\|_{2}^{2}/2$ and the regularized item $\|\bm{ x}\|_{1}$.

In the paper, we  mainly consider the case that $\bm{ x}$ in numerous practical applications is not sparse itself
but is compressible with respect to some given tight frame $\bm{ D}\in\mathbb{R}^{n\times d}~(d\geq n)$.
That is, 
$\bm{ x}=\bm{ D}\bm{ f}$, where $\bm{ f}\in\mathbb{R}^{d}$ is sparse.
Here the Fig. \ref{figure.Original_signal_and_tansformation_coefficents} gives an example, in which the original signal is density  but is sparse under a transformation. Readers can refer to \cite{majumdar2015energy} and find such a real signal example of electroencephalography (EEG) and electrocardiography (ECG).
 And we refer to \cite{han2007frames} for basic theory on tight frame.
  Many  recovery results  for standard sparse signals
  have been extended to the dictionary sparse recovery, see, e.g., \cite{bruckstein2009sparse,candes2011compressed,lin2012new} for $\ell_{1}$ and  \cite{li2014compressed,lin2016restricted} for $\ell_{p}~(0<p<1)$ and references therein.
  Specially, based on the idea of \eqref{VectorL1-Lasso}, 
the following $\ell_1$-analysis problem in \cite{candes2011compressed} is developed:
\begin{equation}\label{VectorL1-ana}
\min_{\bm{ x}\in\mathbb{R}^n}\|\bm{D}^{\top}\bm{ x}\|_1~~\text{subject~ to}
~~\bm{b}-\bm{A}\bm{x}\in\mathcal{B}.
\end{equation}

One of the most commonly used frameworks to investigate the theoretical
performance of the $\ell_{1}$-analysis method is the
restricted isometry property adapted to $\bm{D}$ ($\bm{D}$-RIP), which is first proposed by
Cand{\`e}s et al. \cite{candes2011compressed} and can be defined as follows.

\begin{definition}\label{def.DRIP}
A matrix $\bm{ A}\in\mathbb{R}^{m\times n}$
is said to obey the restricted isometry property adapted to $\bm{ D}\in\mathbb{R}^{n\times d}$
($\bm{D}$-RIP) of order $s$ with constant $\delta_{s}$ if
\begin{equation*}
(1-\delta_{s})\|\bm{Dx}\|_{2}^{2}\leq \|\bm{ A}\bm{Dx}\|_2^2\leq (1+\delta_{s})\|\bm{Dx}\|_{2}^{2}
\end{equation*}
holds for all $s$-sparse vectors $\bm{ v}\in\mathbb{R}^{d}$. The smallest constant $\delta_{s}$ is called as the the restricted isometry constant adapted to $\bm{D}$ ($\bm{D}$-RIC).
\end{definition}

Note that when $\bm{D}$ is an identity matrix, i.e., $\bm{D}=\bm{I}\in\mathbb{R}^{n\times n}$, the  definition of  $\bm{D}$-RIP  reduces to the standard  RIP  in \cite{candes2005decoding,candes2006stable}.

In the recent literature \cite{lin2012new},
the restricted orthogonality constant (ROC) in \cite{candes2005decoding} for the standard
CS  is extended to the CS with  tight frames, which is defined as follows.
\begin{definition}\label{def.DROC}
The $\bm{D}$-restricted orthogonality constant (abbreviated as $\bm{D}$-ROC) of order $(s,t)$, is the
smallest positive number which satisfies
\begin{equation*}
|\langle \bm{ADu}, \bm{ADv}\rangle -\langle \bm{Du}, \bm{Dv}\rangle|\leq \theta_{s,t} \|\bm{u}\|_2\|\bm{v}\|_2
\end{equation*}
for all $s$-sparse and $t$-sparse vectors $\bm{u}, \bm{v}\in\mathbb{R}^{d}$.
\end{definition}

Given the definition $\bm{D}$-RIP and $\bm{D}$-ROC.
For any positive real number $\tau$ , we denote $\theta_{s,\tau t}$  and  $\delta_{\tau t}$ as
$\theta_{s,\lceil \tau t \rceil}$  and  $\delta_{\lceil \tau t\rceil}$, respectively. Here and below, $\lceil t \rceil$ denotes the nearest integer greater than or equal to $t$. With the new notion,
in term of $\bm{D}$-ROC, some sufficient conditions in \cite{lin2012new} have been proposed, such as
$\delta_{s}+1.25\theta_{s,s}<1$, and $\delta_{1.25s}+\theta_{s,1.25s}<1$.

\subsection{MRI Reconstruction Based on Compressed Sensing}\label{s1.2}
\noindent

On the other hand, the idea of sparse signal reconstruction (compressed sensing) also has been used in Magnetic Resonance Imaging (MRI) reconstruction. MRI is crucial in clinical applications for disease diagnosis, since its noninvasive and nonionizing radiation properties enables advanced visualization of anatomical structure.  However, MRI
is limited by the scanning time for physical and physiological constraints \cite{lustig2007sparse}.
In order to reduce the scanning period, the CS method has been introduced into MRI achieving   the capability of accelerating the imaging speed, which is  called CS-MRI technology \cite{lustig2007sparse}.
The success of CS-MRI largely relies on sparsity representation \cite{lustig2008compressed} of the MRI images.
The image is always assumed to be sparse (or compressible) in sparse transform domain (as shown in Fig. \ref{figure.MRI_and_its_transformation}) for the sparsity-based image prior.  Consequently, one can model the reconstruction process through minimizing the regularization function to promote the sparse solution. Sparse representations are confirmed to usually lead to lower reconstruction error \cite{ravishankar2010mr}.

In recent years, some scholars used redundant representation systems instead of orthogonal ones in the MRI community. Some representatives of redundant representations systems are undecimated or shift-invariant wavelet frames
\cite{baker2011translational,guerquin2011fast,liu2015balanced,vasanawala2011practical}, patch-based methods \cite{lai2016image,qu2014magnetic,zhan2015fast}, over-complete dictionaries \cite{ravishankar2010mr,huang2014bayesian}, etc. Redundancy in a redundant system obtains robust image representations and
also introduces additional benefits. For instance, the redundancy in wavelet
enables shift-invariant property. For patch-based methods, redundancy  lead to better noise removal and artifact suppression. The redundant coming from over-complete dictionaries contains more atom signals than required to
represent images, which can better capture different image features, and make the representations
relatively sparser than orthogonal dictionaries do.  Redundancy also makes the designing or training of such dictionaries more flexible. Even though the trained dictionary is orthogonal for image patches, the overall representation system for the whole image maybe still be redundant \cite{zhan2015fast,cai2014data} due to overlapping of patches. Most of these redundant representation systems above can be categorized as tight frame systems \cite{vetterli2012foundations}. More works about CS-MRI, readers can refer to two review papers \cite{sandilya2017compressed,ye2019compressed}.
In this paper, we focus on tight frame-based MRI image reconstruction methods.

\begin{figure*}[htbp!]
\setlength{\tabcolsep}{4.0pt}\small
\begin{tabular}{c}
\includegraphics[width=16.0cm,height=8.0cm]{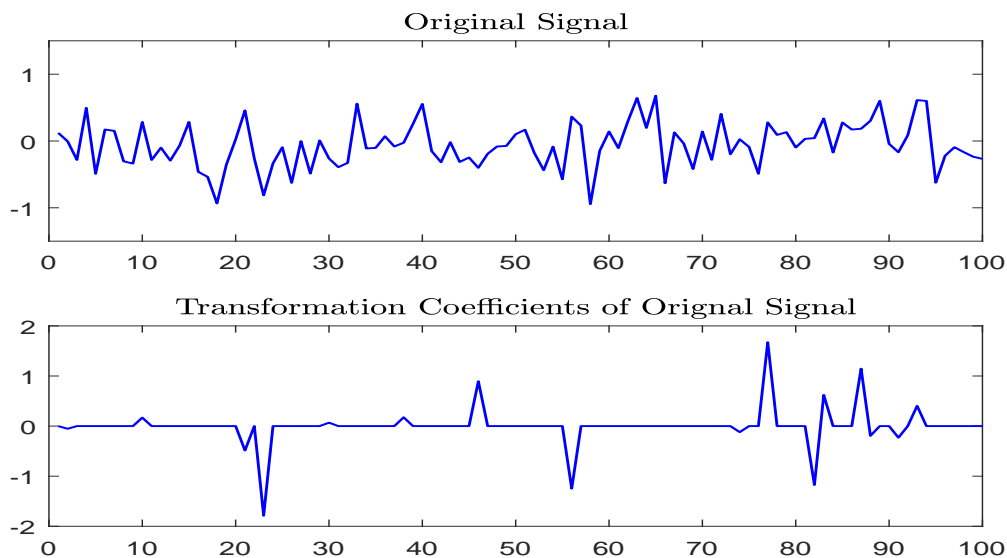}
\end{tabular}
\centering
\caption{\label{figure.Original_signal_and_tansformation_coefficents} Original signal and its transformation coefficients. Upper: the original signal, which is density (not sparse); Down: the coefficients under tight frame transformation, which is sparse. }
\vspace{-0.1cm}
\end{figure*}

\begin{figure*}[htbp!]
\setlength{\tabcolsep}{4.0pt}\small
\begin{tabular}{c}
\includegraphics[width=16.0cm,height=8.0cm]{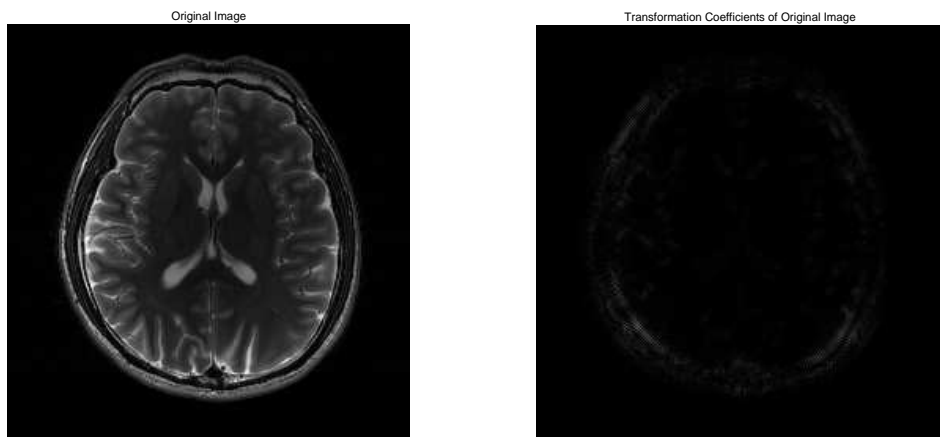}
\end{tabular}
\centering
\caption{\label{figure.MRI_and_its_transformation} Brain MRI and its transformation coefficients. Left: the original Brain MRI, which is density (not sparse); Right: the coefficients under tight frame transformation, which is sparse. Here the black and white colors denote the pixel value $0$ and $1$, respectively.}
\vspace{-0.1cm}
\end{figure*}

\subsection{Contributions}\label{s1.3}
\noindent

In this paper, we  introduce the unconstrained  $\ell_{1}-\alpha\ell_{2}$-analysis  minimization:
\begin{equation}\label{VectorL1-alphaL2-ASSO}
\min_{\bm{x}\in\mathbb{R}^n}~\lambda(\|\bm{ D}^{\top}\bm{x}\|_{1}-\alpha\|\bm{ D}^{\top}\bm{x}\|_{2})
+\frac{1}{2}\|\bm{Ax}-\bm{b}\|_{2}^{2}
\end{equation}
for some constant $0<\alpha\leq 1$, where $\lambda$ is a regularization parameter. Denote \eqref{VectorL1-alphaL2-ASSO} as Analysis $\ell_{1}-\alpha\ell_{2}$-Shrinkage and Selector Operator ($\ell_{1}-\alpha\ell_{2}$-ASSO).
It is rooted to the sparse signal under tight frame and
the constrained and unconstrained  $\|\bm{x}\|_{1}-\alpha\|\bm{ x}\|_{2}$ minimizations, which has recently attracted a lot of attention. The constrained $\|\bm{x}\|_{1}-\alpha\|\bm{x}\|_{2}$ minimization   \cite{li2020minimization,ge2021dantzig,liu2017further,lou2018fast,lou2015computing,yin2015minimization} is
\begin{equation}\label{VectorL1-alphaL2}
\min_{\bm{x}\in\mathbb{R}^n}~\|\bm{x}\|_{1}-\alpha\|\bm{x}\|_{2}
\quad \text{subject \ to} \quad \bm{b}-\bm{A}\bm{x}\in\mathcal{B}.
\end{equation}
And they have been
introduced and studied different  conditions  based on  RIP  for the recovery of $\bm{x}$.
The unconstrained $\|\bm{x}\|_{1}-\alpha\|\bm{x}\|_{2}$ minimization
\cite{liu2017further,lou2018fast,lou2015computing,yin2015minimization,geng2020Unconstrained} is
\begin{equation}\label{VectorL1-L2-SSO}
\min_{\bm{x}\in\mathbb{R}^n}~\lambda(\|\bm{x}\|_{1}-\|\bm{x}\|_{2})+\frac{1}{2}\|\bm{A}\bm{x}-\bm{b}\|_2^2.
\end{equation}
It is a key bridge for finding the solution of the constrained $\|\bm{x}\|_{1}-\alpha\|\bm{x}\|_{2}$ minimization \eqref{systemequationsnoise}.
There is an effective algorithm based on   the different of convex algorithm (DCA) to solve \eqref{VectorL1-L2-SSO}, see
 \cite{lou2015computing,yin2015minimization}. Numerical examples in \cite{li2020minimization,ge2021dantzig,lou2015computing,yin2015minimization}  demonstrate that the $\ell_{1}-\alpha\ell_{2}$ minimization consistently outperforms the $\ell_{1}$ minimization and the $\ell_{p}$ minimization in \cite{lai2013improved} when the measurement matrix $\bm{A}$ is highly coherent.

Motivated by the smoothing and decomposition
transformations in \cite{tan2014smoothing}, the $\ell_{1}-\alpha\ell_{2}$-ASSO  is written as a  general nonsmooth convex optimization problem:
\begin{equation}\label{VectorL1-alphaL2-RASSO}
\min_{\bm{x}\in\mathbb{R}^n}~\lambda(\|\bm{z}\|_{1}-\alpha\|\bm{z}\|_{2})
+\frac{1}{2}\|\bm{Ax}-\bm{b}\|_{2}^{2}+\frac{\rho}{2}\|\bm{D}^{\top}\bm{x}-\bm{z}\|_{2}^{2}.
\end{equation}
We refer to it as Relaxed Analysis $\ell_{1}-\alpha\ell_{2}$-Shrinkage and Selector Operator ($\ell_{1}-\alpha\ell_{2}$-RASSO).

The main  contributions are as follows:
\begin{enumerate}
\item[(i)] We develop a new elementary technique for tight frames (Propositions \ref{prop.DROC} and \ref{NonsparseROC}) and
 show sufficient conditions based on  $\bm{D}$-RIP and $\bm{D}$-ROC
for the recovery of $\bm{x}$ from \eqref{systemequationsnoise} via the  $\ell_{1}-\alpha\ell_{2}$-ASSO \eqref{VectorL1-alphaL2-ASSO}, see to Theorem \ref{StableRecoveryviaVectorL1-alphaL2-ASSO}.

\item[(ii)] We show sufficient conditions based on  $\bm{D}$-RIP frame
for the recovery of the signal $\bm{x}$ from \eqref{systemequationsnoise} via the  $\ell_{1}-\alpha\ell_{2}$-RASSO \eqref{VectorL1-alphaL2-RASSO}, see to Theorem \ref{StableRecoveryviaVectorL1-alphaL2-RASSO}.

\item[(iii)]In order to find the solution of  \eqref{VectorL1-alphaL2-ASSO},
we proposed  an algorithm based on ADMM,  which offers a noticeable boost to some sparse reconstruction algorithms for tight frames.  Numerical examples based the effective algorithm is presented for both sparse signals and MRI.
The recovery performance of the propose method  is highly comparable (often better) to existing state-of-the-art methods.
\end{enumerate}

\subsection{Notations and Organization}\label{s1.4}
\noindent

Throughout the paper, we use the following basic notations.
 For any positive integer $n$, let $[[1,n]]$ be the set $\{1,\ldots,n\}$.
 For the index set $S\subseteq [[1,n]]$, let $|S|$ be  the number of entries in $S$,  $\bm{x}_S$ be the vector equal to $\bm{x}$ on $S$ and to zero on $S^c$, and $\bm{D}_{S}$ be $\bm{D}$ with all but the columns indexed by $S$ set to zero vector.
  For $\bm{x}\in\mathbb{R}^n$,  denote
$\bm{x}_{\max(s)}$ as the vector $\bm{x}$ with all but the largest $s$ entries in absolute value set to zero,
and $\bm{x}_{-\max(s)}=\bm{x}-\bm{x}_{\max(s)}$. Let $\|\bm{D}^{\top}\bm{x}\|_{\alpha,1-2}$ be
$\|\bm{D}^{\top}\bm{x}\|_1-\alpha\|\bm{D}^{\top}\bm{x}\|_2$.
Especially,
when $\alpha=1$, denote $\|\bm{D}^{\top}\bm{x}\|_{\alpha,1-2}$ with $\|\bm{D}^{\top}\bm{x}\|_{1-2}$.
 And we denote $n\times n$ identity matrix by $\bm{I}_{n}$,  zeros matrix by $\bm{ O}$, and the transpose of matrix
 $\bm{A}$ by $\bm{A}^{\top}$. Use the phrase ``$s$-sparse vector" to refer to vectors of sparsity at most $s$. We use boldfaced letter to denote matrix or vector. We denote the standard inner product and Euclidean norm $\langle \cdot,  \cdot\rangle$ and $\|\cdot\|_{2}$, respectively.

The rest of the paper is organized as follows. In Section \ref{adds2}, we describe some preliminaries about the tight frame and the unconstraint $\ell_{1}-\alpha\ell_{2}$ analysis models.  In Section \ref{s3}, we give the main conclusions and show their proof.  Then we design an algorithm based on projected PISTA to solve the proposed model \eqref{VectorL1-alphaL2-ASSO} in Section \ref{s4}.  And we also apply our algorithm to solve the proposed model for signal and MRI reconstruction in Section \ref{s5}. Conclusions are given in Section \ref{s6}.

\section{Preliminaries}\label{adds2}
\noindent

In this section, we first recall some significant lemmas in order to analyse 
the  $\ell_{1}-\alpha\ell_{2}$-ASSO \eqref{VectorL1-alphaL2-ASSO} and  $\ell_{1}-\alpha\ell_{2}$-RASSO \eqref{VectorL1-alphaL2-RASSO}.

In view of the above-mentioned facts in Section \ref{s1}, the $\bm{D}$-RIP plays  a key role for the
stable recovery of signal via the analysis approach. The following lemma show some basic properties in  \cite{lin2012new}  for the $\bm{D}$-RIP with $\delta_{s}$ and $\theta_{s,t}$.
\begin{lemma}\label{Basicproperty}
\begin{enumerate}
\item[(i)]For any $s_{1}\leq s_{2}\leq d$, we have $\delta_{s_{1}}\leq \delta_{s_{2}}$.
\item[(ii)]For any positive integers $s$ and $t_{1}\leq t_{2}$, we have $\theta_{s,t_{2}}\leq\sqrt{\frac{t_{2}}{t_{1}}}\theta_{s,t_{1}}$.
\item[(iii)]$\theta_{s_{1},t_{1}}\leq\theta_{s_{2},t_{2}}$ if  $s_{1}\leq s_{2}$ and  $t_{1}\leq t_{2}$ with $s_{2}+t_{2}\leq d$.
\item[(iv)]For all nonnegative integers $s,t\leq d$,  we have $\theta_{s,t}\leq \delta_{s+t}$.
\end{enumerate}
\end{lemma}

Here, we proposed the standard  properties for the dictionary $\bm{D}$.
\begin{proposition}\label{prop.DROC}
Let $\bar{\bm{D}}\in\mathbb{R}^{(d-n)\times d}$ be the orthogonal complement of $\bm{D}\in\mathbb{R}^{n\times d}$, i.e., $\bar{\bm{D}}\bm{D}^{\top}=\bm{0}$, $\bar{\bm{D}}\bar{\bm{D}}^{\top}=\bm{I}_{n-d}$. Then,
\begin{enumerate}
\item[(i)] For any $\bm{v}\in\mathbb{R}^{d}$, $\|\bm{v}\|_{2}^{2}=\|\bm{Dv}\|_{2}^{2}+\|\bar{\bm{D}}\bm{v}\|_{2}^{2}$.
\item[(ii)] For any $s$-sparse signal $\bm{v}\in\mathbb{R}^{d}$, we have
$$
(1-\delta_{s})\|\bm{v}\|_{2}^{2}\leq \|\bm{A}\bm{Dv}\|_2^2+\|\bar{\bm{D}}\bm{v}\|_{2}^{2}\leq (1+\delta_{s})\|\bm{v}\|_{2}^{2}.
$$
\item[(iii)]$\langle \bm{u}, \bm{v}\rangle=\langle \bar{\bm{D}}\bm{u}, \bar{\bm{D}}\bm{v}\rangle
+\langle \bm{D}\bm{u}, \bm{D}\bm{v}\rangle$ for any $\bm{u}\in\mathbb{R}^{d}$.

\item[(iv)] For any $\|\bm{u}\|_{0}\leq s$ and $\|\bm{v}\|_{0}\leq t$, we have
$$
|\langle \bm{A}\bm{Du}, \bm{A}\bm{Dv}\rangle+\langle\bar{\bm{D}}\bm{u}, \bar{\bm{D}}\bm{v}\rangle
-\langle\bm{u}, \bm{v}\rangle|
\leq\theta_{s,t}\|\bm{u}\|_{2}\|\bm{v}\|_{2}.
$$
Moreover, if $\text{\rm supp}(\bm{u})\cap\text{\rm supp}(\bm{v})=\emptyset$, we have
$$
|\langle \bm{A}\bm{Du}, \bm{A}\bm{Dv}\rangle+\langle\bar{\bm{D}}\bm{u}, \bar{\bm{D}}\bm{v}\rangle|
\leq\theta_{s,t}\|\bm{u}\|_{2}\|\bm{v}\|_{2}.
$$
\end{enumerate}
\end{proposition}

\begin{proof}
Despite the results $(i)$ and $(ii)$ in Proposition \ref{prop.DROC} are clear,  we here give their proofs.
As far as we know,   the results (iii) and (iv) are first proved.

$(i)$
For any $v\in\mathbb{R}^{d}$, there is
\begin{align*}
\|\bm{Dv}\|_{2}^{2}+\|\bar{\bm{D}}\bm{v}\|_{2}^{2}
&=\langle \bm{D}^{\top} \bm{D}\bm{v}, \bm{v}\rangle +\langle \bar{\bm{D}}^{\top}\bar{\bm{D}}\bm{v},\bm{v}\rangle
\nonumber\\
&=\bm{v}^{\top}\begin{bmatrix}
     \bm{D}^{\top}  &  \bar{\bm{D}}^{\top} \\
  \end{bmatrix}
\begin{bmatrix}
     \bm{D} \\
    \bar{\bm{D}} \\
  \end{bmatrix}\bm{v}
=\|\bm{v}\|_{2}^{2},
\end{align*}
where the last equality comes from the definition of tight frame $\bm{D}\bm{D}^{\top}=\bm{I}_{n}$ and the definition of the orthogonal complement of frame $\bar{\bm{D}}\bar{\bm{D}}^{\top}=\bm{I}_{n-d}$.

$(ii)$ 
By the definition of $\bm{D}$-RIP, we have
\begin{align}\label{prop.DROC.eq6}
\|\bm{A}\bm{Dv}\|_2^2+\|\bar{\bm{D}}\bm{v}\|_{2}^{2}
&\leq (1+\delta_{s})\|\bm{Dv}\|_{2}^{2}+\|\bar{\bm{D}}\bm{v}\|_{2}^{2}
\leq (1+\delta_{s})\left(\|\bm{Dv}\|_{2}^{2}+\|\bar{\bm{D}}\bm{v}\|_{2}^{2}\right)\nonumber\\
&=(1+\delta_{s})\|\bm{v}\|_{2}^{2},
\end{align}
where the equality comes from the item $(i)$. Similarly, we can show the lower bound
\begin{eqnarray}\label{prop.DROC.eq7}
\|\bm{A}\bm{Dv}\|_{2}^{2}+\|\bar{\bm{D}}\bm{v}\|_{2}^{2}
\geq(1-\delta_{s})\|\bm{v}\|_{2}^{2}.
\end{eqnarray}
The combination of the upper bound \eqref{prop.DROC.eq6} and the lower bound \eqref{prop.DROC.eq7} gives the item (ii).

Next, we show our new conclusions items (iii) and (iv).

$(iii)$ 
According to the parallelogram identity, we have
\begin{align*}
\langle \bar{\bm{D}}\bm{u}, \bar{\bm{D}}\bm{v}\rangle
&=\frac{1}{4}\left(\|\bar{\bm{D}}(\bm{u}+\bm{v})\|_{2}^{2}
-\|\bar{\bm{D}}(\bm{u}-\bm{v})\|_{2}^{2}\right),\nonumber\\
\langle \bm{D}\bm{u}, \bm{D}\bm{v}\rangle
&=\frac{1}{4}\left(\|\bm{D}(\bm{u}+\bm{v})\|_{2}^{2}
-\|\bm{D}(\bm{u}-\bm{v})\|_{2}^{2}\right).
\end{align*}
Taking sum over both side of above two equalities, one has
\begin{align}\label{prop.DROC.eq5}
&\langle \bar{\bm{D}}\bm{u}, \bar{\bm{D}}\bm{v}\rangle
+\langle \bm{D}\bm{u}, \bm{D}\bm{v}\rangle\nonumber\\
&=\frac{1}{4}\left(\left(\|\bar{\bm{D}}(\bm{u}+\bm{v})\|_{2}^{2}+\|\bm{D}(\bm{u}+\bm{v})\|_{2}^{2}\right)
-\left(\|\bar{\bm{D}}(\bm{u}-\bm{v})\|_{2}^{2}+\|\bm{D}(\bm{u}-\bm{v})\|_{2}^{2}\right)\right)\nonumber\\
&=\frac{1}{4}\left(\|\bm{u}+\bm{v}\|_{2}^{2}-\|\bm{u}-\bm{v}\|_{2}^{2}\right),
\end{align}
where the second equality holds because of the item $(i)$.  Last, by the above \eqref{prop.DROC.eq5} and the following parallelogram identity
\begin{equation*}
\langle \bm{u}, \bm{v}\rangle
=\frac{1}{4}\left(\|(\bm{u}+\bm{v})\|_{2}^{2}
-\|(\bm{u}-\bm{v})\|_{2}^{2}\right),
\end{equation*}
we get
\begin{equation*}
\langle \bm{u}, \bm{v}\rangle=\langle \bar{\bm{D}}\bm{u}, \bar{\bm{D}}\bm{v}\rangle
+\langle \bm{D}\bm{u}, \bm{D}\bm{v}\rangle.
\end{equation*}

$(iv)$ 
It follows from the item $(ii)$ and the definition of  $\bm{D}$-ROC that
\begin{eqnarray*}
&&|\langle \bm{A}\bm{Du}, \bm{A}\bm{Dv}\rangle+\langle\bar{\bm{D}}\bm{u}, \bar{\bm{D}}\bm{v}\rangle
-\langle\bm{u}, \bm{v}\rangle|
=|\langle \bm{A}\bm{Du}, \bm{A}\bm{Dv}\rangle-\langle\bm{D}\bm{u}, \bm{D}\bm{v}\rangle|nonumber\\
&&\leq\theta_{s,t}\|\bm{u}\|_{2}\|\bm{v}\|_{2},
\end{eqnarray*}
which is our desired  conclusion.
\end{proof}

Next, the following inequality is a modified cone
constraint inequality for the unconstrained $\|\bm{D}^{\top}\bm{x}\|_{1}-\alpha\|\bm{D}^{\top}\bm{x}\|_{2}$ minimization \eqref{VectorL1-alphaL2-ASSO}. Its proof is similar as that of \cite[Lemma 2.3]{geng2020Unconstrained} and we omit its proof.

\begin{lemma}\label{ConeTubeconstraint-L1-L2-ASSO}
Let  $\bm{b}=\bm{Ax}+\bm{e}$ with
$\|\bm{e}\|_{2}\leq\lambda$, $T = \text{\rm supp}(\bm{D}^{\top}\bm {x}_{\max(t)})$, and $\hat{\bm{x}}$ be a
minimization solution of \eqref{VectorL1-alphaL2-ASSO}. Then
\begin{align}\label{ConeTubeconstraint-L1-L2-ASSO.eq1}
\|\bm{Ah}\|_{2}^{2}+2\lambda\|\bm{D}_{T^c}^{\top}\bm{h}\|_{1}\leq2\lambda(\|\bm{D}_{T}^{\top}\bm{h}\|_1
+\alpha\|\bm{D}^{\top}\bm{h}\|_2+2\|\bm{D}_{T^c}^{\top}\bm{x}\|_{1}
+\|\bm{Ah}\|_{2}),
\end{align}
where ${\bm h}=\hat{{\bm x}}-{\bm x}$.  And moreover,
\begin{align*}
\|\bm{Ah}\|_2^2+2\lambda(\|{\bm D}_{T^c}^{\top}{\bm h}\|_1-\alpha\|{\bm D}_{T^c}^{\top}{\bm h}\|_2)
\leq2\lambda(\|{\bm D}_{T}^{\top}{\bm h}\|_1+2\|{\bm D}_{T^c}^{\top}{\bm x}\|_1
+\alpha\|{\bm D}_{T}^{\top}{\bm h}\|_2+\|\bm {Ah}\|_2),
\end{align*}
\begin{align*}
\|{\bm D}_{T^c}^{\top}{\bm h}\|_1\leq\|{\bm D}_{T}^{\top}{\bm h}\|_1+2\|{\bm D}_{T^c}^{\top}{\bm x}\|_1
+\alpha\|{\bm D}^{\top}{\bm h}\|_2+\|\bm {Ah}\|_2,
\end{align*}
and
\begin{align*}
\|{\bm Ah}\|_2^2\leq2\lambda(\|{\bm D}_{T}^{\top}{\bm h}\|_1+2\|{\bm D}_{T^c}^{\top}{\bm x}\|_1
+\alpha\|{\bm D}^{\top}{\bm h}\|_2+\|\bm {Ah}\|_2).
\end{align*}
\end{lemma}

For the $\ell_{1}-\alpha\ell_{2}$-RASSO  \eqref{VectorL1-alphaL2-RASSO}, there is a similar cone
constraint inequality.
\begin{lemma}\label{ConeTubeconstraint-L1-L2-RASSO}
Let  ${\bm b}={\bm {Ax}}+{\bm e}$  with
$\|\bm e\|_2\leq\lambda$, and  $T = \text{\rm supp}({\bm D}^{\top}{\bm x}_{\max(s)})$. Then the
minimization solution $\hat{\bm x}$ of \eqref{VectorL1-alphaL2-RASSO} satisfies
\begin{align}\label{ConeTubeconstraint-L1-L2-RASSO.eq1}
\|\bm {Ah}\|_2^2+2\lambda\|{\bm D}_{T^c}^{\top}{\bm h}\|_1\leq2\lambda\bigg(\|{\bm D}_{T}^{\top}{\bm h}\|_1+\alpha\|{\bm D}^{\top}{\bm h}\|_2+2\|{\bm D}_{T^c}^{\top}{\bm x}\|_1
+\|\bm {Ah}\|_2+\frac{(\alpha+1)^2d}{2}\frac{\lambda}{\rho}\bigg),
\end{align}
where ${\bm h}=\hat{{\bm x}}-{\bm x}$.
\end{lemma}

\begin{proof}
Our proof is motivated by the proof of \cite[Lemma 2.3]{geng2020Unconstrained} and \cite[Lemma IV.4]{tan2014smoothing}.
Assume that $(\hat{\bm{x}},\hat{\bm{z}})$ is another solution of \eqref{VectorL1-alphaL2-RASSO}, then
\begin{align}\label{ConeTubeconstraint-L1-L2-RASSO.eq2}
0&\geq \bigg(\lambda(\|\hat{{\bm z}}\|_{1}-\alpha\|\hat{{\bm z}}\|_{2})+\frac{1}{2}\|{\bm  A}\hat{{\bm  x}}-{\bm  b}\|_{2}^{2}+\frac{\rho}{2}\|{\bm D}^{\top}\hat{{\bm x}}-\hat{{\bm z}}\|_{2}^{2}\bigg)\nonumber\\
&\hspace{12pt}-\bigg(\lambda(\|{\bm D}^{\top}{\bm x}\|_{1}-\alpha\|{\bm D}^{\top}{\bm x}\|_{2})+\frac{1}{2}\|{\bm  A}{\bm  x}-{\bm  b}\|_{2}^{2}\bigg)\nonumber\\
&=\frac{1}{2}\bigg(\|{\bm  A}\hat{{\bm  x}}-{\bm  b}\|_{2}^{2}-\|{\bm  A}{\bm  x}-{\bm  b}\|_{2}^{2}\bigg)\nonumber\\
&\hspace{12pt}+\bigg(\lambda(\|\hat{{\bm z}}\|_{1}-\alpha\|\hat{{\bm z}}\|_{2})+\frac{\rho}{2}\|{\bm D}^{\top}\hat{{\bm x}}-\hat{{\bm z}}\|_{2}^{2}-\lambda(\|{\bm D}^{\top}{\bm x}\|_{1}-\alpha\|{\bm D}^{\top}{\bm x}\|_{2})\bigg)\nonumber\\
&=:I_1+I_2.
\end{align}

Next, we deal with terms $I_1$ and $I_2$, respectively. By $\hat{{\bm x}}={\bm h}+{\bm x}$ and ${\bm Ax}-{\bm b}=-{\bm e}$,
\begin{align}\label{ConeTubeconstraint-L1-L2-RASSO.eq3}
I_1&=\frac{1}{2}\bigg(\|{\bm  A}{\bm  h}-{\bm e}\|_{2}^{2}-\|-{\bm  e}\|_{2}^{2}\bigg)
=\frac{1}{2}\|{\bm  A}{\bm  h}\|_{2}^{2}-\langle {\bm  A}{\bm  h}, {\bm e}\rangle \nonumber\\
&\overset{(1)}{\geq} \frac{1}{2}\|{\bm  A}{\bm  h}\|_{2}^{2}-\|{\bm  A}{\bm  h}\|_{2}\|{\bm e}\|_{2}
\overset{(2)}{\geq} \frac{1}{2}\|{\bm  A}{\bm  h}\|_{2}^{2}-\lambda\|{\bm  A}{\bm  h}\|_{2},
\end{align}
where the inequality (1) comes from the Cauchy-Schwarz inequality, and  the inequality (2) is from $\|{\bm e}\|_{2}\leq \lambda$.

Then, we estimate the term $I_{2}$. The subgradient optimality condition for \eqref{VectorL1-alphaL2-RASSO} is
\begin{align*}
{\bm 0}&={\bm A}^{\top}({\bm A}\hat{{\bm x}}-{\bm b})+\rho{\bm D}({\bm D}^{\top}\hat{{\bm x}}-\hat{{\bm z}}),\nonumber\\
{\bm 0}&\in\lambda{\bm v}-\rho({\bm D}^{\top}\hat{{\bm x}}-\hat{{\bm z}}),
\end{align*}
where ${\bm v}$ is a Frech$\acute{e}$t  sub-differential  of the function $\|\hat{{\bm z}}\|_{1}-\alpha\|\hat{{\bm z}}\|_{2}$.  Then, we derive that
\begin{align}\label{ConeTubeconstraint-L1-L2-RASSO.eq4}
I_2&= \lambda(\|\hat{{\bm z}}\|_{1}-\|{\bm D}^{\top}{\bm x}\|_{1})-\lambda\alpha(\|\hat{{\bm z}}\|_{2}-\|{\bm D}^{\top}{\bm x}\|_{2})
+\frac{\rho}{2}\|{\bm D}^{\top}\hat{{\bm x}}-\hat{{\bm z}}\|_{2}^{2}\nonumber\\
&= \lambda\left(\left\|{\bm D}^{\top}\hat{{\bm x}}-\frac{\lambda}{\rho}{\bm v}\right\|_{1}-\|{\bm D}^{\top}{\bm x}\|_{1}\right)
-\lambda\alpha\left(\left\|{\bm D}^{\top}\hat{{\bm x}}-\frac{\lambda}{\rho}{\bm v}\right\|_{2}-\|{\bm D}^{\top}{\bm x}\|_{2}\right)
+\frac{\rho}{2}\|\frac{\lambda}{\rho}{\bm v}\|_{2}^{2}\nonumber\\
&\overset{(a)}{\geq}\lambda\left(\|{\bm D}^{\top}\hat{{\bm x}}\|_1-\|{\bm D}^{\top}{\bm x}\|_1 -\frac{\lambda}{\rho}\|{\bm v}\|_{1}\right)
-\lambda\alpha\left(\|{\bm D}^{\top}\hat{{\bm x}}\|_2-\|{\bm D}^{\top}{\bm x}\|_2 +\frac{\lambda}{\rho}\|{\bm v}\|_{2} \right)
+\frac{\lambda^2}{2\rho}\|{\bm v}\|_{2}^{2}\nonumber\\
&=\lambda\left(\left(\|{\bm D}^{\top}\hat{{\bm x}}\|_1-\|{\bm D}^{\top}{\bm x}\|_1\right)
-\alpha\left(\|{\bm D}^{\top}\hat{{\bm x}}\|_2-\|{\bm D}^{\top}{\bm x}\|_2\right)\right)\nonumber\\
&\hspace*{12pt}+\frac{\lambda^2}{2\rho}(\|{\bm v}\|_{2}^{2}-2\|{\bm v}\|_{1}-2\alpha\|{\bm v}\|_{2})\nonumber\\
&=:I_{21}+I_{22},
\end{align}
where (a)  is because of $\|{\bm u}\|_{p}-\|{\bm v}\|_{p}\leq \|\bm{u +v}\|_{p}\leq \|{\bm u}\|_{p}+\|{\bm v}\|_{p}$ for any $1\leq p\leq\infty$. Note that
\begin{align}\label{ConeTubeconstraint-L1-L2-RASSO.eq5}
I_{21}&=\lambda\left(\|{\bm D}_{T}^{\top}(\bm{ h+ x})+{\bm D}_{T^c}^{\top}(\bm{ h+ x})\|_1-\|{\bm D}_{T}^{\top}{\bm x}
+{\bm D}_{T^c}^{\top}{\bm x}\|_1\right)\nonumber\\
&\hspace*{12pt}-\lambda\alpha\left(\|{\bm D}^{\top}(\bm{ h+ x})\|_2-\|{\bm D}^{\top}{\bm x}\|_2\right)\nonumber\\
&\geq\lambda\left(\|{\bm D}_{T}^{\top}{\bm x}\|_1-\|{\bm D}_{T}^{\top}{\bm h}\|_1+\|{\bm D}_{T^c}^{\top}{\bm h}\|_1-\|{\bm D}_{T^c}^{\top}{\bm x}\|_1-
\|{\bm D}_{T}^{\top}{\bm x}\|_1-\|{\bm D}_{T^c}^{\top}{\bm x}\|_1\right)\nonumber\\
&\hspace*{12pt}-\lambda\alpha\left(\|{\bm D}^{\top}{\bm h}\|_2\right)\nonumber\\
&=\lambda(\|{\bm D}_{T^c}^{\top}{\bm h}\|_1-\|{\bm D}_{T}^{\top}{\bm h}\|_1-2\|{\bm D}_{T^c}^{\top}{\bm x}\|_1-\alpha\|{\bm D}^{\top}{\bm h}\|_2),
\end{align}
where the inequality is because of $\|\bm{u +v}\|_{1}\geq \|{\bm u}\|_{1}-\|{\bm v}\|_{1}$ and $\|\bm{u +v}\|_{1}\leq\|{\bm u}\|_{1}+\|{\bm v}\|_{1}$. On the other hand,
\begin{align}\label{ConeTubeconstraint-L1-L2-RASSO.eq6}
I_{22}&\overset{(1)}{\geq}\frac{\lambda^2}{2\rho}\left(\|{\bm v}\|_{2}^{2}-2(1+\alpha)\|{\bm v}\|_{1}\right)\nonumber\\
&=\frac{\lambda^2}{2\rho}\sum_{j=1}^{d}\left((|v|_j-(1+\alpha))^2-(1+\alpha)^2\right)
\overset{(2)}{\geq} -\frac{(1+\alpha)^2\lambda^2d}{2\rho},
\end{align}
where (1) is because of $\|{\bm v}\|_{2}\leq \|{\bm v}\|_{1}$, and the equality (2) holds if and only if $|v|_j=(1+\alpha)$ for all $j=1,\ldots,d$.   Substituting the estimate \eqref{ConeTubeconstraint-L1-L2-RASSO.eq5} and \eqref{ConeTubeconstraint-L1-L2-RASSO.eq6} into \eqref{ConeTubeconstraint-L1-L2-RASSO.eq4}, one has
\begin{align}\label{ConeTubeconstraint-L1-L2-RASSO.eq7}
I_2\geq\lambda\left(\|{\bm D}_{T^c}^{\top}{\bm h}\|_1-\|{\bm D}_{T}^{\top}{\bm h}\|_1-2\|{\bm D}_{T^c}^{\top}{\bm x}\|_1-\alpha\|{\bm D}^{\top}{\bm h}\|_2-\frac{(1+\alpha)^2d\lambda }{2\rho}\right).
\end{align}

Last, substituting the estimate of $I_1$ in \eqref{ConeTubeconstraint-L1-L2-RASSO.eq3} and $I_2$ in \eqref{ConeTubeconstraint-L1-L2-RASSO.eq7} into the inequality \eqref{ConeTubeconstraint-L1-L2-RASSO.eq2}, we get \eqref{ConeTubeconstraint-L1-L2-RASSO.eq1}.
\end{proof}

The following lemma is the fundamental properties of the metric  $\|{\bm  x}\|_1-\alpha\|{\bm  x}\|_2$ with
$0\leq\alpha\leq 1$, see \cite{yin2015minimization} and  \cite{ge2021dantzig}.
\begin{lemma}\label{LocalEstimateL1-L2}
For any $\bm {x}\in \mathbb{R}^n$, the following statements hold:
\begin{description}
  \item (a) Let $0 \leq \alpha\leq 1$, $T=\text{\rm supp}({\bm  x})$ and $\|{\bm  x}\|_0=s$. Then
\begin{align}\label{e:l12alphas}
(s-\alpha\sqrt{s})\min_{j\in T}|x_j|\leq\|{\bm  x}\|_1-\alpha\|{\bm  x}\|_2\leq(\sqrt{s}-\alpha)\|{\bm  x}\|_2.
\end{align}
\item (b) Let $S, S_1, S_2\subseteq [n]$ satisfy $S=S_1\cup S_2$ and $S_1\cap S_2=\emptyset$. Then
\begin{align}\label{e:uplowerbound}
\|\bm {x}_{S_1}\|_{1}-\alpha\|\bm{x}_{S_1}\|_{2}+\|\bm{x}_{S_2}\|_{1}-\alpha\|\bm{x}_{S_2}\|_{2}\leq\|\bm{x}_{S}\|_{1}-
\alpha\|\bm{x}_{S}\|_{2}.
\end{align}
\end{description}
\end{lemma}

Now, we can give the key technical tool used in the main results.
It provides an estimation  based on $\bm{D}$-ROC for $|\langle \bm{ ADu}, \bm{ ADv} \rangle +\langle \bar{{\bm D}}{\bm u},\bar{{\bm D}}{\bm v}\rangle|$, where
one of $\bm{u}$ or $\bm{v}$ is sparse.
Our idea is inspired by \cite[Lemma 5.1]{cai2013compressed} and \cite[Lemma 2.3]{li2019signal}.

\begin{proposition}\label{NonsparseROC}
Let $s_1, s_2\leq n$. Suppose ${\bm u},{\bm v}\in\mathbb{R}^d$ satisfy $\text{\rm supp}({\bm u})\cap \text{\rm supp}({\bm v})=\emptyset$ and $\bm u$ is $s_1$ sparse. If $\|{\bm v}\|_{1-2}\leq(s_2-\sqrt{s_2})\eta$ and $\|{\bm v}\|_\infty\leq\eta$, then
\begin{align}\label{NonsparseROC.eq1}
|\langle \bm{ ADu}, \bm{ ADv} \rangle +\langle \bar{{\bm D}}{\bm u},\bar{{\bm D}}{\bm v}\rangle|
\leq\bigg(1+\frac{\sqrt{2}}{2}\bigg)\eta \sqrt{s_2}\theta_{s_1,s_2}\|\bm{u}\|_2.
\end{align}
\end{proposition}

Before giving the proof of Proposition \ref{NonsparseROC},  we first recall
a convex combination of sparse vectors for any point based on  the metric $\|\cdot\|_1-\|\cdot\|_2$.

\begin{lemma}\label{e:convexl1-2}\cite[Lemma2.2]{ge2021new}
Let a vector $\bm {\nu}\in \mathbb{R}^d$ satisfy $\|\bm {\nu}\|_{\infty}\leq \eta$, where $\alpha$ is a positive constant.
Suppose $\|\bm {\nu}\|_{1-2}\leq (s-\sqrt{s})\eta$
 with   a positive integer $s$ and $s\leq |\mathrm{supp}(\bm {\nu})|$.  Then $\bm {\nu}$ can be represented as a convex
combination of $s$-sparse vectors $\bm {v}^{(i)}$, i.e.,
\begin{align}\label{e:com}
\bm {\nu}=\sum_{i=1}^N\lambda_i\bm {v}^{(i)},
\end{align}
where $N$ is a positive integer,
 \begin{align}\label{e:thetaisum}
 0<\lambda_i\leq1,\ \ \ \  \sum_{i=1}^N\lambda_i=1,
 \end{align}
\begin{align}\label{e:setS}
\mathrm{supp}(\bm {v}^{(i)})\subseteq \mathrm{supp}(\bm {\nu}), \ \ \|\bm {v}^{(i)}\|_0\leq {s},\ \ \|\bm {v}^{(i)}\|_{\infty}\leq \Big(1+\frac{\sqrt{2}}{2}\Big)\eta,
\end{align}
and
\begin{align}\label{e:newinequality}
\sum_{i=1}^N\lambda_i\|\bm {v}^{(i)}\|_2^2
 \leq \Big[\Big(1+\frac{\sqrt{2}}{2}\Big)^2(s-\sqrt{s})+1\Big]\theta^2.
\end{align}
\end{lemma}

Now, we give the proof of Proposition \ref{NonsparseROC} in details.
\begin{proof}
Suppose $\|\bm{v}\|_0=t$. We consider two cases as follows.

\textbf{Case I: $t\leq s_2$.}

By the item (iv) of Proposition \ref{prop.DROC} and $\|{\bm v}\|_\infty\leq\alpha$, we have
\begin{align}\label{NonsparseROC.eq2}
|\langle \bm{ ADu}, \bm{ ADv} \rangle +\langle \bar{{\bm D}}{\bm u},\bar{{\bm D}}{\bm v}\rangle|&\leq\theta_{s_1,t}\|{\bm u}\|_2\|{\bm v}\|_2\leq\theta_{s_1,t}\|{\bm u}\|_2\|{\bm v}\|_{\infty}\sqrt{\|{\bm v}\|_0}\nonumber\\
&\leq\eta\sqrt{s_2}\theta_{s_1,s_2}\|{\bm u}\|_2.
\end{align}

\textbf{Case II: $t> s_2$.}

We shall prove by induction. Assume that \eqref{NonsparseROC.eq1} holds for $t-1$.
Note that $\|{\bm v}\|_{1-2}\leq\eta (s_2-\sqrt{s_2})$ and $\|{\bm v}\|_{\infty}\leq \eta$.  By Lemma \ref{e:convexl1-2}, ${\bm v}$ can be represented as the convex hull of $s_2$-sparse vectors:
$$
{\bm v}=\sum_{j=1}^N\gamma_j{\bm v}^j,
$$
where ${\bm v}^j$ is $s_2$-sparse for any $j\in[N]$ and
\begin{align*}
\sum_{j=1}^N\gamma_j=1,~~0<\gamma_j\leq 1, j\in[N].
\end{align*}

Since $\bm{v}^j$ is $s_2$-sparse and $s_2\leq t-1$, we  use the induction assumption,
\begin{align}\label{NonsparseROC.eq3}
&|\langle \bm{ ADu}, \bm{ ADv} \rangle +\langle \bar{{\bm D}}{\bm u},\bar{{\bm D}}{\bm v}\rangle|
\leq\sum_{j=1}^N\gamma_j|\langle \bm{ ADu}, \bm{ AD}\bm{v}^j\rangle
+\langle \bar{{\bm D}}{\bm u},\bar{{\bm D}}{\bm v}^j\rangle|\nonumber\\
&\overset{(a)}{\leq} \sum_{j=1}^N\gamma_j\bigg(\theta_{s_1,s_2}\|\bm{v}^j\|_2\|\bm{u}\|_2\bigg)
\leq \sum_{j=1}^N\gamma_j\bigg(\theta_{s_1,s_2}\|\bm{v}^j\|_{\infty}\sqrt{\|\bm{v}^j\|_{0}}\|\bm{u}\|_2\bigg)\nonumber\\
&\overset{(b)}{\leq}\sum_{j=1}^N\gamma_j\bigg(\bigg(1+\frac{\sqrt{2}}{2}\bigg)\eta\sqrt{s_2}\theta_{s_1,s_2}\|\bm{u}\|_2\bigg)
=\bigg(1+\frac{\sqrt{2}}{2}\bigg)\eta\sqrt{s_2}\theta_{s_1,s_2}\|\bm{u}\|_2,
\end{align}
where (a)  is in view of the item (iv) of Proposition \ref{prop.DROC}, and (b)  is in virtue of Lemma \ref{e:convexl1-2}.

The combination of \textbf{Case I} and \textbf{Case II} is \eqref{NonsparseROC.eq1}.
\end{proof}

\section{Main Result based on  D-RIC and  D-ROC }\label{s3}
\noindent

In this section, we develop sufficient conditions
based on  $\bm{D}$-RIC and  $\bm{D}$-ROC of the
$\ell_{1}-\alpha \ell_2$-ASSO \eqref{VectorL1-alphaL2-ASSO} and $\ell_{1}-\alpha \ell_2$-RASSO \eqref{VectorL1-alphaL2-RASSO}
for the signal recovery applying the technique of the convex combination.

\subsection{Auxiliary Lemmas Under RIP Frame}\label{s3.1}
\noindent

Combining the above auxiliary results in Section \ref{adds2},
we first introduce some main inequalities, which play an important role in establishing the recovery condition of the
$\ell_{1}-\alpha \ell_2$-ASSO \eqref{VectorL1-alphaL2-ASSO} and $\ell_{1}-\alpha \ell_2$-RASSO \eqref{VectorL1-alphaL2-RASSO}.

\begin{proposition}\label{lem:CrossItem}
For positive integers  $s\geq 2$ and $k\geq 1$,
let $S = \text{\rm supp}({\bm D}^{\top}{\bm h}_{\max(s)})$ and $T = \text{\rm supp}({\bm D}^{\top}{\bm x}_{\max(s)})$.
Assume that
\begin{eqnarray}\label{lem:CrossItem.eq1}
\|\bm{D}_{S^c}^{\top}\bm{h}\|_1-\alpha\|\bm{D}_{S^c}^{\top}\bm{h}\|_2
\leq a\|\bm{D}_{S}^{\top}\bm{h}\|_1+b\|\bm{D}_{S}^{\top}\bm{h}\|_2+c\|\bm{D}_{T^c}^{\top}\bm{x}\|_1+\eta\|{\bm {Ah}}\|_2
+\gamma,
\end{eqnarray}
where $a>0$ and $b,c, \eta,\gamma\geq0$ satisfy $(a-1)\sqrt{s}+(b+1)\geq0$. Then the following statements hold:

\begin{description}
  \item[(i)] For the set $S$, one has
   \begin{align}\label{e:CrossItem1}
&|\langle \bm{AD}\bm{D}_{S}^{\top}{\bm h},\bm{AD}\bm{D}_{S^c}^{\top}{\bm h}\rangle
|\nonumber\\
&\leq \theta_{s,s}\sqrt{s}\Big(1+\frac{\sqrt{2}}{2}\Big)\Big(\frac{a\sqrt{s}+b}{\sqrt{s}-1}\frac{\|\bm{D}_{S}^{\top}\bm{h}\|_2}{\sqrt{s}}
+\frac{c\|\bm{D}_{T^c}^{\top}\bm{x}\|_1+\eta\|{\bm Ah}\|_2+\gamma}{s-\sqrt{s}}\Big)\|\bm{D}_{S}^{\top}{\bm h}\|_2,
\end{align}
  \item
   \item[(ii)] Let the set
   \begin{align}\label{def:S}
\tilde{S}=S\cup \Big\{i:|(\bm{D}_{S^c}^{\top}\bm{h})(i)|>\frac{1}{t-1}\frac{a\sqrt{s}+b}{\sqrt{s}-1}\frac{\|\bm{D}_{S}^{\top}\bm{h}\|_2}{\sqrt{s}}
+\frac{c\|\bm{D}_{T^c}^{\top}\bm{x}\|_1+\eta\|{\bm Ah}\|_2+\gamma}{(t-1)(s-\sqrt{s})}\Big\}
\end{align}
for  any $t>3$. Then
   \begin{align}\label{e:CrossItem2}
&|\langle \bm{AD}\bm{D}_{\tilde{S}}^{\top}{\bm h},\bm{AD}\bm{D}_{\tilde{S}^c}^{\top}{\bm h}\rangle
|\leq\theta_{ts, (t-1)s} \frac{\sqrt{\lceil(t-1)s\rceil}}{t-1}  \Big(1+\frac{\sqrt{2}}{2}\Big)
\nonumber\\
&\times\Big(\frac{a\sqrt{s}+b}{\sqrt{s}-1}\frac{\|\bm{D}_{S}^{\top}\bm{h}\|_2}{\sqrt{s}}
+\frac{c\|\bm{D}_{T^c}^{\top}\bm{x}\|_1+\eta\|{\bm Ah}\|_2+\gamma}{s-\sqrt{s}}\Big)\|\bm{D}_{\tilde{S}}^{\top}{\bm h}\|_2.
\end{align}

\item[(iii)]
    Let  $\bm {A}$ satisfy $\bm D$-RIP with
\begin{align}\label{RIPConditiona1add}
\rho_{s,t}=:\delta_{ts}+\sqrt{\frac{\lceil(t-1)s\rceil}{(t-1)^2s}}\frac{(\sqrt{2}+1)(\sqrt{s}+\alpha)}{\sqrt{2}(\sqrt{s}-1)}
\theta_{st, (t-1)s}<1,
\end{align}
with  $t\geq3$. Then
\begin{align}\label{prop:NSP.eq1}
\|\bm{D}_{S}^{\top}\bm{h}\|_2\leq&\frac{(\sqrt{2}+1)\theta_{ts,(t-1)s}}{\sqrt{2}(1-\rho_{s,t})}\frac{\sqrt{\lceil(t-1)s\rceil}}{(t-1)(s-\sqrt{s})}\Big(c\|\bm{D}_{T^c}^{\top}\bm{x}\|_1+\gamma\Big)
\nonumber\\
&+
\Big(\frac{(\sqrt{2}+1)\theta_{ts,(t-1)s}}{\sqrt{2}(1-\rho_{s,t})}\frac{\sqrt{(t-1)s}}{(t-1)(s-\sqrt{s})}+\frac{\sqrt{1+\delta_{ts}}}{1-\rho_{s,t}}\Big)\eta\|\bm{Ah}\|_{2},
\end{align}
where $\eta\geq 1$.

\item[(iv)]
Let  $\bm {A}$ satisfy $\bm D$-RIP with
\begin{align}\label{RIPCondition1}
\rho_{s}=:\delta_{s}+\frac{(\sqrt{2}+1)(\sqrt{s}+\alpha)}{\sqrt{2}(\sqrt{s}-1)}
\theta_{s,s}<1.
\end{align}
Then
\begin{align}\label{prop:NSP.eq2}
\|\bm{D}_{S}^{\top}\bm{h}\|_2\leq&\frac{(\sqrt{2}+1)\theta_{s,s}}{\sqrt{2}(1-\rho_{s})(\sqrt{s}-1)}\Big(c\|\bm{D}_{T^c}^{\top}\bm{x}\|_1+\gamma\Big)
\nonumber\\
&+
\Big(\frac{(\sqrt{2}+1)\theta_{s,s}}{\sqrt{2}(1-\rho_{s})(\sqrt{s}-1)}+\frac{\sqrt{1+\delta_{s}}}{1-\rho_{s}}\Big)\eta\|\bm{Ah}\|_{2},
\end{align}
where $\eta\geq 1$.

\item[(v)] For the term $\|\bm{D}_{S^c}^{\top}\bm{h}\|_2$, there is
\begin{align}\label{DSC}
&\|\bm{D}_{S^c}^{\top}\bm{h}\|_2\nonumber\\
&\leq\Bigg(\sqrt{\frac{a\sqrt{s}+b}{\sqrt{s}}+\frac{\alpha^2}{4s}}+\frac{\alpha+\bar{\varepsilon}}{2\sqrt{s}}\Bigg)\|\bm{D}_{S}^{\top}\bm{h}\|_2
+\frac{1}{2\bar{\varepsilon}}\big(c\|\bm{D}_{T^c}^{\top}\bm{x}\|_1+\eta\|\bm{Ah}\|_2+\gamma\big),
\end{align}
where $\bar{\varepsilon}>0$ is a constant.
\end{description}
\end{proposition}

\begin{proof}
Please see Appendix \ref{appendx1}.
\end{proof}

\subsection{Main Result Under RIP Frame }\label{s3.2}
\noindent

Now, we show  the stable recovery conditions based  $\bm{D}$-RIP for
the  $\ell_{1}-\alpha \ell_{2}$-ASSO (\ref{VectorL1-alphaL2-ASSO}) and
the  $\ell_{1}-\alpha \ell_{2}$-RASSO (\ref{VectorL1-alphaL2-RASSO}), respectively.
\begin{theorem}\label{StableRecoveryviaVectorL1-alphaL2-ASSO}
Consider $\bm { b}=\bm { Ax}+\bm {e}$ with $\|\bm {e}\|_{2}\leq \eta$.
  Let  $\hat{\bm {x}}$ be the minimizer of the  $\ell_{1}-\alpha \ell_{2}$-ASSO (\ref{VectorL1-alphaL2-ASSO}).
 The following statements hold:
 \begin{description}
     \item[(i)]
    If the  measurement matrix $\bm {A}$ satisfies \eqref{RIPConditiona1add},
  then
  \begin{align*}
&\|\hat{\bm {x}}-\bm {x}\|_2\nonumber\\
\leq&
\Bigg(\left(\sqrt{\frac{\alpha+\sqrt{s}}{\sqrt{s}}+\frac{\alpha^2}{4s}}+\frac{\alpha+1}{2\sqrt{s}}+1
\right)\left(\tau+\frac{\left(\tau(1-\rho_{s,t})+\sqrt{1+\delta_{ts}}\right)C}{C(1-\rho_{s,t})+(\sqrt{s}+\alpha)\sqrt{1+\delta_{ts}}}\right)\nonumber\\
&\hspace*{12pt}+\frac{1}{2}\left(1+\frac{C(1-\rho_{s,t})}{C(1-\rho_{s,t})+(\sqrt{s}+\alpha)\sqrt{1+\delta_{ts}}}\right)\Bigg)
2\|\bm{D}_{T^c}^{\top}\bm{x}\|_1\nonumber\\
&+\left(\left(\sqrt{\frac{\alpha+\sqrt{s}}{\sqrt{s}}+\frac{\alpha^2}{4s}}+\frac{\alpha+1}{2\sqrt{s}}+1
\right)\frac{\tau(1-\rho_{s,t})+\sqrt{1+\delta_{ts}}}{1-\rho_{s,t}}+\frac{1}{2}\right)\nonumber\\
&\hspace*{12pt}\times\left(C+\frac{(\sqrt{s}+\alpha)\sqrt{1+\delta_{ts}}}{(1-\rho_{s,t})}\right)2\lambda,
\end{align*}
where the constants $\tau$ and $C$ are as follows
\begin{align}\label{Constant.t>=3}
\begin{cases}
\tau=\frac{(\sqrt{2}+1)\theta_{ts,(t-1)s}}{\sqrt{2}(1-\rho_{s,t})}\frac{\sqrt{(t-1)s}}{(t-1)(s-\sqrt{s})},&\\
C=1+\frac{(\sqrt{2}+1)\theta_{ts,(t-1)s}}{\sqrt{2}(1-\rho_{s,t})}\frac{(\sqrt{s}+\alpha)}{(s-\sqrt{s})}\frac{\sqrt{(t-1)s}}{t-1}.
\end{cases}
\end{align}

    \item[(ii)]If the  measurement matrix $\bm {A}$ satisfies
  \eqref{RIPCondition1}, then
      \begin{align*}
\|\hat{\bm {x}}-\bm {x}\|_2
\leq&
\Bigg(\left(\sqrt{\frac{\alpha+\sqrt{s}}{\sqrt{s}}+\frac{\alpha^2}{4s}}+\frac{\alpha+1}{2\sqrt{s}}+1
\right)\left(\tau+\frac{\left(\tau(1-\rho_{s})+\sqrt{1+\delta_{s}}\right)C}{C(1-\rho_{s})+(\sqrt{s}+\alpha)\sqrt{1+\delta_{s}}}\right)\nonumber\\
&\hspace*{12pt}+\frac{1}{2}\left(1+\frac{C(1-\rho_{s})}{C(1-\rho_{s})+(\sqrt{s}+\alpha)\sqrt{1+\delta_{s}}}\right)\Bigg)
2\|\bm{D}_{T^c}^{\top}\bm{x}\|_1\nonumber\\
&+\left(\left(\sqrt{\frac{\alpha+\sqrt{s}}{\sqrt{s}}+\frac{\alpha^2}{4s}}+\frac{\alpha+1}{2\sqrt{s}}+1
\right)\frac{\tau(1-\rho_{s})+\sqrt{1+\delta_{s}}}{1-\rho_{s}}+\frac{1}{2}\right)\nonumber\\
&\hspace*{12pt}\times\left(C+\frac{(\sqrt{s}+\alpha)\sqrt{1+\delta_{s}}}{(1-\rho_{s})}\right)2\lambda,
\end{align*}
where the constants $\tau$ and $C$ are as follows
\begin{align}\label{Constant.t=2}
\begin{cases}
\tau=\frac{(\sqrt{2}+1)\theta_{s,s}}{\sqrt{2}(1-\rho_{s})}\frac{1}{\sqrt{s}-1},&\\
C=1+\frac{(\sqrt{2}+1)\theta_{s,s}}{\sqrt{2}(1-\rho_{s})}\frac{\sqrt{s}+\alpha}{\sqrt{s}-1}.
\end{cases}
\end{align}
  \end{description}

\end{theorem}

\begin{theorem}\label{StableRecoveryviaVectorL1-alphaL2-RASSO}
Consider $\bm { b}=\bm { Ax}+\bm {e}$ with $\|\bm {e}\|_{2}\leq \eta$. Let $\hat{\bm {x}}$ be the minimizer
of the  $\ell_{1}-\alpha \ell_{2}$-RASSO (\ref{VectorL1-alphaL2-RASSO}).
The following statements hold:
\begin{description}
  \item[(i)] If the  measurement matrix $\bm {A}$ satisfies the $\bm D$-RIP  with \eqref{RIPConditiona1add}, then
  \begin{align*}
&\|\hat{\bm {x}}-\bm {x}\|_2\nonumber\\
\leq&
\Bigg(\left(\sqrt{\frac{\alpha+\sqrt{s}}{\sqrt{s}}+\frac{\alpha^2}{4s}}+\frac{\alpha+1}{2\sqrt{s}}+1
\right)\left(\tau+\frac{\left(\tau(1-\rho_{s,t})+\sqrt{1+\delta_{ts}}\right)C}{C(1-\rho_{s,t})+(\sqrt{s}+\alpha)\sqrt{1+\delta_{ts}}}\right)\nonumber\\
&\hspace*{12pt}+\frac{1}{2}\left(1+\frac{C(1-\rho_{s,t})}{C(1-\rho_{s,t})+(\sqrt{s}+\alpha)\sqrt{1+\delta_{ts}}}\right)\Bigg)\nonumber\\
&\times\left(2\|\bm{D}_{T^c}^{\top}\bm{x}\|_1+\frac{(\alpha+1)^2d}{2}\frac{\lambda}{\rho}\right)\nonumber\\
&+\left(\left(\sqrt{\frac{\alpha+\sqrt{s}}{\sqrt{s}}+\frac{\alpha^2}{4s}}+\frac{\alpha+1}{2\sqrt{s}}+1
\right)\frac{\tau(1-\rho_{s,t})+\sqrt{1+\delta_{ts}}}{1-\rho_{s,t}}+\frac{1}{2}\right)\nonumber\\
&\hspace*{12pt}\times\left(C+\frac{(\sqrt{s}+\alpha)\sqrt{1+\delta_{ts}}}{(1-\rho_{s,t})}\right)2\lambda,
\end{align*}
where the constants $\tau$ and $C$ are in \eqref{Constant.t>=3}.

  \item[(ii)] If the  measurement matrix $\bm {A}$ satisfies the $\bm D$-RIP  with \eqref{RIPCondition1}, then
  \begin{align*}
\|\hat{\bm {x}}-\bm {x}\|_2
\leq&
\Bigg(\left(\sqrt{\frac{\alpha+\sqrt{s}}{\sqrt{s}}+\frac{\alpha^2}{4s}}+\frac{\alpha+1}{2\sqrt{s}}+1
\right)\left(\tau+\frac{\left(\tau(1-\rho_{s})+\sqrt{1+\delta_{s}}\right)C}{C(1-\rho_{s})+(\sqrt{s}+\alpha)\sqrt{1+\delta_{s}}}\right)\nonumber\\
&\hspace*{12pt}+\frac{1}{2}\left(1+\frac{C(1-\rho_{s})}{C(1-\rho_{s})+(\sqrt{s}+\alpha)\sqrt{1+\delta_{s}}}\right)\Bigg)\nonumber\\
&\times\left(2\|\bm{D}_{T^c}^{\top}\bm{x}\|_1+\frac{(\alpha+1)^2d}{2}\frac{\lambda}{\rho}\right)\nonumber\\
&+\left(\left(\sqrt{\frac{\alpha+\sqrt{s}}{\sqrt{s}}+\frac{\alpha^2}{4s}}+\frac{\alpha+1}{2\sqrt{s}}+1
\right)\frac{\tau(1-\rho_{s})+\sqrt{1+\delta_{s}}}{1-\rho_{s}}+\frac{1}{2}\right)\nonumber\\
&\hspace*{12pt}\times\left(C+\frac{(\sqrt{s}+\alpha)\sqrt{1+\delta_{s}}}{(1-\rho_{s})}\right)2\lambda,
\end{align*}
where the constants $\tau$ and $C$ are in \eqref{Constant.t=2}.
\end{description}

\end{theorem}

\begin{remark}
The condition \eqref{RIPCondition1} 
reduces to
\begin{equation}\label{RIPCondition2}
\delta_{2s}<\frac{1}{\frac{\bigg(1+\frac{\sqrt{2}}{2}\bigg)(\sqrt{s}+\alpha)}{\sqrt{s}-1}+1}.
\end{equation}
It is clearly  weaker than the following  condition in \cite[Theorem 2]{ge2021dantzig}
\begin{equation*}
\delta_{2s}<\frac{1}{\sqrt{1+\frac{(\sqrt{s}+{\alpha})^2\Big(\Big(1+\frac{\sqrt{2}}{2}\Big)^2( s-\sqrt{s})+1\Big)}{s(\sqrt{s}-1)^2}}},
\end{equation*}
and
the following  condition in \cite[Theorem 3.4]{ge2021new}
\begin{equation*}
\delta_{2s}<\frac{1}{\sqrt{1+\frac{(\sqrt{s}+1)^2\Big(\Big(1+\frac{\sqrt{2}}{2}\Big)^2( s-\sqrt{s})+1\Big)}{s(\sqrt{s}-1)^2}}}.
\end{equation*}
\end{remark}

\begin{remark}
We notice that Wang and Wang \cite[Equation (6)]{wang2019improved} established the following condition
\begin{equation}\label{RIPCondition.wang2019improved}
\delta_{s}+\frac{\sqrt{s}+\sqrt{2}-1}{\sqrt{s}-1}
\theta_{s,s}<1.
\end{equation}
Though it is weaker than our condition 
\begin{equation}\label{RIPCondition3}
\delta_{s}+\frac{(\sqrt{2}+1)(\sqrt{s}+\alpha)}{\sqrt{2}(\sqrt{s}-1)}
\theta_{s,s}<1
\end{equation}
for {$s\geq2$}, their condition is for the constraint $\ell_1-\ell_2$ model rather than unconstraint $\ell_1-\ell_2$ model.
\end{remark}

\begin{remark}
Our condition \eqref{RIPCondition2} is weaker than that of \cite[Corollary 2]{geng2020Unconstrained}
\begin{equation}\label{CoherenceCondition.geng2020Unconstrained}
\mu<\frac{1}{3s+6},
\end{equation}
owing to $\mu=\delta_{2s}$.
\end{remark}

Proofs of Theorems $1$ and $2$ is similar, so we only present
the detail proof of Theorem \ref{StableRecoveryviaVectorL1-alphaL2-RASSO}.

\begin{proof}[Proof of Theorem \ref{StableRecoveryviaVectorL1-alphaL2-RASSO}]
We first show the conclusion $(\bm{i})$. By $\bm{D}\bm{D}^{\top}=\bm{I}_n$, we know
\begin{align}\label{decomposition}
\|\bm{h}\|_2=\|\bm{D}^{\top}\bm{h}\|_2=\sqrt{\|\bm{D}_{S}^{\top}\bm{h}\|_2^2+\|\bm{D}_{S^c}^{\top}\bm{h}\|_2^2}
\leq \|\bm{D}_{S}^{\top}\bm{h}\|_2+ \|\bm{D}_{S^c}^{\top}\bm{h}\|_2,
\end{align}
where $\|\bm{D}_{S}^{\top}\bm{h}\|_2$ and $\|\bm{D}_{S^c}^{\top}\bm{h}\|_2$ are needed to estimated, respectively.

We first estimate $\|\bm{D}_{S}^{\top}\bm{h}\|_2$.
 By Lemma \ref{ConeTubeconstraint-L1-L2-RASSO}, then
 the condition \eqref{lem:CrossItem.eq1} holds with $a=1,~b=\alpha,~c=2,~\eta=1,~\gamma=(\alpha+1)^{2}d\lambda/(2\rho)$.
 Using Proposition \ref{lem:CrossItem} $(iii)$, one has
 \begin{align}\label{prop:NSP.eq1cge}
\|\bm{D}_{S}^{\top}\bm{h}\|_2\leq&\frac{(\sqrt{2}+1)\theta_{ts,(t-1)s}}{\sqrt{2}(1-\rho_{s,t})}
\frac{\sqrt{\lceil(t-1)s\rceil}}{(t-1)(s-\sqrt{s})}
\Big(2\|\bm{D}_{T^c}^{\top}\bm{x}\|_1+\frac{(\alpha+1)^{2}d\lambda}{2\rho}\Big)
\nonumber\\
&+
\Big(\frac{(\sqrt{2}+1)\theta_{ts,(t-1)s}}{\sqrt{2}(1-\rho_{s,t})}\frac{\sqrt{\lceil(t-1)s\rceil}}
{(t-1)(s-\sqrt{s})}+\frac{\sqrt{1+\delta_{ts}}}{1-\rho_{s,t}}\Big)\|\bm{Ah}\|_{2}\nonumber\\
=&
\tau
\Big(2\|\bm{D}_{T^c}^{\top}\bm{x}\|_1+\frac{(\alpha+1)^{2}d\lambda}{2\rho}\Big)
+\bar{\tau}\|\bm{Ah}\|_{2},
\end{align}
where
\begin{align}\label{tau}
\begin{cases}
\tau=\frac{(\sqrt{2}+1)\theta_{ts,(t-1)s}}{\sqrt{2}(1-\rho_{s,t})}\frac{\sqrt{\lceil(t-1)s\rceil}}{(t-1)(s-\sqrt{s})},&\\
\bar{\tau}=\tau+\frac{\sqrt{1+\delta_{ts}}}{1-\rho_{s,t}}.
\end{cases}
\end{align}

In order to estimate $\|\bm{D}_{S}^{\top}\bm{h}\|_2$, we need an upper bound of $\|\bm{Ah}\|_{2}$. Using Lemma \ref{ConeTubeconstraint-L1-L2-RASSO}, we derive
\begin{align*}
&\|\bm{Ah}\|_2^2-2\lambda\|\bm{Ah}\|_2\nonumber\\
&\leq2\lambda\left(\|\bm{D}_{S}^{\top}\bm{h}\|_1+2\|\bm{D}_{T^c}^{\top}\bm{x}\|_1+\frac{(\alpha+1)^2d\lambda}{2\rho}+\alpha\|\bm{D}^{\top }\bm{h}\|_2-\|\bm{D}_{S^c}^{\top}\bm{h}\|_1\right)\nonumber\\
&\leq2\sqrt{s}\lambda\|\bm{D}_{S}^{\top}\bm{h}\|_2
+2\lambda\left(2\|\bm{D}_{T^c}^{\top}\bm{x}\|_1+\frac{(\alpha+1)^2d\lambda}{2\rho}+\alpha\|\bm{D}_{S}^{\top}\bm{h}\|_2+\alpha\|\bm{D}_{S^c}^{\top }\bm{h}\|_2-\|\bm{D}_{S^c}^{\top}\bm{h}\|_1\right)\nonumber\\
&\overset{(a)}{\leq} 2\lambda(\sqrt{s}+\alpha)\|\bm{D}_{S}^{\top}\bm{h}\|_2+2\lambda\left(2\|\bm{D}_{T^c}^{\top }\bm{x}\|_1+\frac{(\alpha+1)^2d\lambda}{2\rho}\right)\nonumber\\
&\overset{(b)}{\leq}2\lambda(\sqrt{s}+\alpha)\left(\bar{\tau}\|\bm{Ah}\|_2+\tau\left(2\|\bm{D}_{T^c}^{\top }\bm{x}\|_1+\frac{(\alpha+1)^2d\lambda}{2\rho}\right)\right)\nonumber\\
&\hspace*{12pt}+2\lambda\left(2\|\bm{D}_{T^c}^{\top}\bm{x}\|_1+\frac{(\alpha+1)^2d\lambda}{2\rho}\right)\nonumber\\
&=2\lambda(\sqrt{s}+\alpha)\bar{\tau}\|\bm{Ah}\|_2+2\lambda(1+(\sqrt{s}+\alpha)\tau)\left(2\|\bm{D}_{T^c}^{\top }\bm{x}\|_1+\frac{(\alpha+1)^2d\lambda}{2\rho}\right),
\end{align*}
where (a) follows from $\alpha\|\bm{D}_{S^c}^{\top}\bm{h}\|_2\leq\|\bm{D}_{S^c}^{\top}\bm{h}\|_2\leq\|\bm{D}_{S^c}^{\top}\bm{h}\|_1$,
(b) is because of \eqref{prop:NSP.eq1cge}. Thus
\begin{align*}
\|\bm{Ah}\|_2^2-2\lambda(1+(\sqrt{s}+\alpha)\bar{\tau})\|\bm{Ah}\|_2-2\lambda(1+(\sqrt{s}+\alpha)\tau )\left(2\|\bm{D}_{T^c}^{\top }\bm{x}\|_1+\frac{(\alpha+1)^2d\lambda}{2\rho}\right)\leq0.
\end{align*}
By the fact that the second order inequality $aX^2-bX-c\leq0$ for $a,b,c>0$ has the solution
$$X\leq \frac{b+\sqrt{b^2+4ac}}{2a}\leq\frac{b+\sqrt{(b+2ac/b)^2}}{2a}
=\frac{b}{a}+\frac{c}{b},$$
then
\begin{align}\label{Ah}
\|\bm{Ah}\|_2
&\leq(1+(\sqrt{s}+\alpha)\bar{\tau})2\lambda+\frac{(1+(\sqrt{s}+\alpha)\tau )\left(2\|\bm{D}_{T^c}^{\top }\bm{x}\|_1+\frac{(\alpha+1)^2d\lambda}{2\rho}\right)}{1+(\sqrt{s}+\alpha)\bar{\tau}}\nonumber\\
&=2 \bar{C}\lambda+\frac{C}{\bar{C}}\Big(2\|\bm{D}_{T^c}^{\top}\bm{x}\|_1+\frac{(\alpha+1)^2d\lambda}{2\rho}\Big),
\end{align}
where
\begin{align}\label{C.constant}
\begin{cases}
C=1+(\sqrt{s}+\alpha)\tau
=1+\frac{(\sqrt{2}+1)\theta_{ts,(t-1)s}}{\sqrt{2}(1-\rho_{s,t})}\frac{(\sqrt{s}+\alpha)}{(s-\sqrt{s})}\frac{\sqrt{(t-1)s}}{t-1},&\\
\bar{C}=1+(\sqrt{s}+\alpha)\bar{\tau}=C + \frac{(\sqrt{s}+\alpha)\sqrt{1+\delta_{ts}}}{1-\rho_{s,t}}.
\end{cases}
\end{align}

Combining \eqref{prop:NSP.eq1cge} with \eqref{Ah}, one has
\begin{eqnarray}\label{StableRecoveryviaVectorL1-alphaL2-ASSO.eq1}
&&\|\bm{D}_{S}^{\top}\bm{h}\|_2\nonumber\\
&&\leq\tau \left(2\|\bm{D}_{T^c}^{\top}\bm{x}\|_1+\frac{(\alpha+1)^2d}{2}\frac{\lambda}{\rho}\right)+\bar{\tau}
\left(2\lambda \bar{C}+\frac{C}{\bar{C}}\Big(2\|\bm{D}_{T^c}^{\top}\bm{x}\|_1+\frac{(\alpha+1)^2d\lambda}{2\rho}\Big)\right)\nonumber\\
&&=\left(\tau +\frac{C}{\bar{C}}\bar{\tau}\right)\left(2\|\bm{D}_{T^c}^{\top}\bm{x}\|_1+\frac{(\alpha+1)^2d}{2}\frac{\lambda}{\rho}\right)
+2\bar{\tau}\bar{C}\lambda.
\end{eqnarray}

Next, we consider an upper bound of $\|\bm{D}_{S^c}^{\top}\bm{h}\|_2$.
And from Proposition \eqref{lem:CrossItem} $(v)$ with $\bar{\varepsilon}=1$, it follows that
\begin{align}\label{DSCge}
&\|\bm{D}_{S^c}^{\top}\bm{h}\|_2\nonumber\\
&\leq\Bigg(\sqrt{\frac{\sqrt{s}+\alpha}{\sqrt{s}}
+\frac{\alpha^2}{4s}}+\frac{\alpha+1}{2\sqrt{s}}\Bigg)\|\bm{D}_{S}^{\top}\bm{h}\|_2
+\frac{1}{2}\Big(2\|\bm{D}_{T^c}^{\top}\bm{x}\|_1+\|\bm{Ah}\|_2+\frac{(\alpha+1)^{2}d\lambda}{2\rho}\Big).
\end{align}
Similarly, combining \eqref{Ah} with \eqref{StableRecoveryviaVectorL1-alphaL2-ASSO.eq1}, \eqref{DSCge} reduces to
\begin{align}\label{StableRecoveryviaVectorL1-alphaL2-ASSO.eq2}
&\|\bm{D}_{S^c}^{\top}\bm{h}\|_2
\leq\Bigg(\sqrt{\frac{\alpha+\sqrt{s}}{\sqrt{s}}+\frac{\alpha^2}{4s}}+\frac{\alpha+1}{2\sqrt{s}}\Bigg)\|\bm{D}_{S}^{\top}\bm{h}\|_2
+\frac{1}{2}\left(2\|\bm{D}_{T^c}^{\top}\bm{x}\|_1+\frac{(\alpha+1)^2d}{2}\frac{\lambda}{\rho}\right)\nonumber\\
&\hspace*{12pt}+\frac{1}{2}\Big(2\bar{C} \lambda +\frac{C}{\bar{C}}\Big(2\|\bm{D}_{T^c}^{\top }\bm{x}\|_1+\frac{(\alpha+1)^2d\lambda}{2\rho}\Big)\Big)\nonumber\\
&\leq\Bigg(\sqrt{\frac{\alpha+\sqrt{s}}{\sqrt{s}}+\frac{\alpha^2}{4s}}+\frac{\alpha+1}{2\sqrt{s}}
\Bigg)\|\bm{D}_{S}^{\top}\bm{h}\|_2
+\frac{1}{2}\bigg(1+\frac{C}{\bar{C}}\bigg)\left(2\|\bm{D}_{T^c}^{\top}\bm{x}\|_1+\frac{(\alpha+1)^2d}{2}\frac{\lambda}{\rho}\right)\nonumber\\
&\hspace*{12pt}+\bar{C}\lambda.
\end{align}

Therefore, substituting the estimation of $\|\bm{D}_{S}^{\top}\bm{h}\|_2$ and $\|\bm{D}_{S^c}^{\top}\bm{h}\|_2$  into \eqref{decomposition}, one has
\begin{align*}
\|\bm {h}\|_2
&\overset{(a)}{\leq}\Bigg(\sqrt{\frac{\alpha+\sqrt{s}}{\sqrt{s}}+\frac{\alpha^2}{4s}}+\frac{\alpha+1}{2\sqrt{s}}+1
\Bigg)\|\bm{D}_{S}^{\top}\bm{h}\|_2\nonumber\\
&\hspace*{12pt}+\frac{1}{2}\bigg(1+\frac{C}{\bar{C}}\bigg)\left(2\|\bm{D}_{T^c}^{\top}\bm{x}\|_1+\frac{(\alpha+1)^2d}{2}\frac{\lambda}{\rho}\right)
+\bar{C}\lambda\nonumber\\
&\overset{(b)}{\leq}\Bigg(\sqrt{\frac{\alpha+\sqrt{s}}{\sqrt{s}}+\frac{\alpha^2}{4s}}+\frac{\alpha+1}{2\sqrt{s}}+1
\Bigg)\left(\left(\tau +\bar{\tau}\frac{C}{\bar{C}}\right)\left(2\|\bm{D}_{T^c}^{\top}\bm{x}\|_1+\frac{(\alpha+1)^2d}{2}\frac{\lambda}{\rho}\right)
+\bar{\tau}\bar{C}2\lambda\right)\nonumber\\
&\hspace{24pt}+\frac{1}{2}\bigg(1+\frac{C}{\bar{C}}\bigg)\left(2\|\bm{D}_{T^c}^{\top}\bm{x}\|_1+\frac{(\alpha+1)^2d}{2}\frac{\lambda}{\rho}\right)
+\bar{C}\lambda\nonumber\\
&=
\Bigg(\Bigg(\sqrt{\frac{\alpha+\sqrt{s}}{\sqrt{s}}+\frac{\alpha^2}{4s}}+\frac{\alpha+1}{2\sqrt{s}}+1
\Bigg)\left(\tau +\bar{\tau}\frac{C}{\bar{C}}\right)+\frac{1}{2}\bigg(1+\frac{C}{\bar{C}}\bigg)\Bigg)\nonumber\\
&\hspace*{12pt}\times\left(2\|\bm{D}_{T^c}^{\top}\bm{x}\|_1+\frac{(\alpha+1)^2d}{2}\frac{\lambda}{\rho}\right)
+\left(\left(\sqrt{\frac{\alpha+\sqrt{s}}{\sqrt{s}}+\frac{\alpha^2}{4s}}+\frac{\alpha+1}{2\sqrt{s}}+1
\right)\bar{\tau}+\frac{1}{2}\right)2\bar{C}\lambda\nonumber\\
&=:E_{1}\left(2\|\bm{D}_{T^c}^{\top}\bm{x}\|_1+\frac{(\alpha+1)^2d}{2}\frac{\lambda}{\rho}\right)
+E_{2}2\lambda,
\end{align*}
where (a) is due to \eqref{StableRecoveryviaVectorL1-alphaL2-ASSO.eq2},  and (b) is from \eqref{StableRecoveryviaVectorL1-alphaL2-ASSO.eq1}.
Recall the definition of the constants $\tau,\bar{\tau}$ in \eqref{tau} and $C,\bar{C}$ in \eqref{C.constant},
we get
\begin{align}
&E_{1}
=\Bigg(\sqrt{\frac{\alpha+\sqrt{s}}{\sqrt{s}}+\frac{\alpha^2}{4s}}+\frac{\alpha+1}{2\sqrt{s}}+1\Bigg)\nonumber\\
&\hspace*{12pt}\times\left(\tau +\left(\tau+\frac{\sqrt{1+\delta_{ts}}}{1-\rho_{s,t}}\right)\left(\frac{C(1-\rho_{s,t})}{C(1-\rho_{s,t})+(\sqrt{s}+\alpha)\sqrt{1+\delta_{ts}}}\right)\right)
\nonumber\\
&+\frac{1}{2}\bigg(1+\frac{C(1-\rho_{s,t})}{C(1-\rho_{s,t})+(\sqrt{s}+\alpha)\sqrt{1+\delta_{ts}}}\bigg)\nonumber\\
&=\Bigg(\sqrt{\frac{\alpha+\sqrt{s}}{\sqrt{s}}+\frac{\alpha^2}{4s}}+\frac{\alpha+1}{2\sqrt{s}}+1
\Bigg)\left(\tau +\frac{\big(\tau(1-\rho_{s,t})+\sqrt{1+\delta_{ts}}\big)C}{C(1-\rho_{s,t})+(\sqrt{s}+\alpha)\sqrt{1+\delta_{ts}}}\right)\nonumber\\
&+\frac{1}{2}\bigg(1+\frac{C(1-\rho_{s,t})}{C(1-\rho_{s,t})+(\sqrt{s}+\alpha)\sqrt{1+\delta_{ts}}}\bigg),\nonumber\\
E_{2}&=\left(\left(\sqrt{\frac{\alpha+\sqrt{s}}{\sqrt{s}}+\frac{\alpha^2}{4s}}+\frac{\alpha+1}{2\sqrt{s}}+1
\right)\frac{\tau(1-\rho_{s,t})+\sqrt{1+\delta_{ts}}}{1-\rho_{s,t}}+\frac{1}{2}\right)\nonumber\\
&\hspace*{12pt}\times\left(C+\frac{(\sqrt{s}+\alpha)\sqrt{1+\delta_{ts}}}{1-\rho_{s,t}}\right).
\end{align}
Therefore, we complete the proof of item ($i$).

($ii$) We can prove the conclusion ($ii$) in a similar way only by replacing Proposition \ref{lem:CrossItem} (iii) with Proposition \ref{lem:CrossItem} (iv).

\end{proof}

\section{Numerical Algorithm }\label{s4}
\noindent

In this section, we develop an efficient algorithm to solve the $\ell_1-\alpha\ell_2$-ASSO
\eqref{VectorL1-alphaL2-ASSO}.

The projected fast iterative soft-thresholding algorithm (pFISTA) for tight frames  in  \cite{liu2016projected} is
suited solving the $\ell_1$-analysis problem \eqref{VectorL1-ana} for 
 MRI reconstruction.
 Compared to the common iterative reconstruction methods such as iterative soft-thresholding algorithm (ISTA) in \cite{daubechies2004iterative} and  fast iterative soft-thresholding algorithm (FISTA) in \cite{beck2009fast},
the pFISTA algorithm  achieves better reconstruction and  converges faster.
Inspired by the pFISTA algorithm,
 an efficient algorithm is  introduced to solve the nonconvex-ASSO problem \eqref{VectorL1-alphaL2-ASSO} in this section.

For the given frame $\bm{D}\in\mathbb{R}^{n\times d}$,
there are many dual frames.
Here, we only consider its canonical dual frame
\begin{equation}\label{canonicaldual}
\bm{\Phi}=(\bm{D}\bm{D}^{\top})^{-1}\bm{D},
\end{equation}
which  satisfies
$$
\bm{\Phi}\bm{D}^{\top}=\bm{I}_{n},
$$
and  is also the pseudo-inverse of $\bm{D}$ \cite{elad2007analysis}.
Taking $\bm{z}=\bm{D}^{\top}\bm{x}$,
the $\ell_1-\alpha\ell_2$-ASSO
\eqref{VectorL1-alphaL2-ASSO} is written as
\begin{equation}\label{VectorL1-alphaL2-ASSO.equi1}
\min_{\bm{ z}\in\text{Range}(\bm{D}^{\top})}~\lambda(\|\bm{z}\|_{1}-\alpha\|\bm{z}\|_{2})+\frac{1}{2}\|\bm{ A}\bm{\Phi}\bm{  z}-\bm{ b}\|_{2}^{2}.
\end{equation}

Now, we  solve  \eqref{VectorL1-alphaL2-ASSO.equi1} by the idea of  the pFISTA algorithm.
 First, we introduce an indicator function
\begin{equation*}
\chi(\bf{z})=
\begin{cases}
\bm{0},~\bm{z}\in\text{Range}(\bm{D}^{\top}),&\\
+\bm{\infty},~\text{otherwise},
\end{cases}
\end{equation*}
then  an equivalent unconstrained model of \eqref{VectorL1-alphaL2-ASSO.equi1} is
\begin{equation}\label{VectorL1-alphaL2-ASSO.equi2}
\min_{\bm{ z}\in\mathbb{R}^{d}}~\lambda(\|\bm{ z}\|_{1}-\alpha\|\bm{ z}\|_{2})+\chi(\bm{z})+\frac{1}{2}\|\bm{ A}\bm{\Phi}\bm{ z}-\bm{ b}\|_{2}^{2}.
\end{equation}
Let
$$
g(\bm{z})=\lambda(\|\bm{ z}\|_{1}-\alpha\|\bm{ z}\|_{2})+\chi(\bm{z}),~~h(\bm{z})=\frac{1}{2}\|{\bm  A}\bm{\Phi}{\bm  z}-{\bm  b}\|_{2}^{2},
$$
then \eqref{VectorL1-alphaL2-ASSO.equi2} can be rewritten as
\begin{equation}\label{OptimizationProblem}
\min_{{\bm  z}\in\mathbb{R}^{d}}~g(\bm{z})+h(\bm{z}),
\end{equation}
where $g$ is a non-smooth function, and $h$ is a smooth
function with a $\ell_{h}$-Lipschitz continuous gradient ($\ell_{h}>0$), i.e
$$
\|\nabla h(\bm{z}_1)- \nabla h(\bm{z}_2)\|_{2}\leq\ell_{h}\|\bm{z}_1-\bm{z}_2\|_{2}.
$$

Next, we solve \eqref{VectorL1-alphaL2-ASSO.equi2} via ISTA by incorporating the proximal mapping
\begin{align}\label{ISTA1}
\bm{z}^{k+1}&=\text{Prox}_{\gamma g}(\bm{z}^{k}-\gamma \nabla h(\bm{z}^{k}))\nonumber\\
&=\arg\min_{\bm{ z}\in\text{Range}(\bm{D}^{\top})}\gamma\lambda(\|\bm{ z}\|_{1}-\alpha\|\bm{z}\|_{2})
+\frac{1}{2}\left\|\bm{z}-\left(\bm{z}^{k}-\gamma \nabla h(\bm{z}^{k})\right)\right\|_{2}^{2},
\end{align}
where $\gamma$ is the step size
and $\text{Prox}_{\gamma g}(\cdot)$ is the proximal operator of
the function $\gamma g$.
The  proximal operator of $\mu_1\ell_1-\mu_2\ell_2$ in \cite[Proposition 7.1]{liu2017further} and \cite[Section 2]{lou2018fast}
is
\begin{equation}\label{L1L2-ProximalMap4}
\text{Prox}_{\lambda(\ell_1-\alpha\ell_2)}({\bm  b})=\arg\min_{\bm{x}}\frac{1}{2}\|\bm{x}-\bm{b}\|_2^2+\lambda(\|\bm{  x}\|_1-\alpha\|\bm{ x}\|_2), \ \ \ \ 0<\alpha\leq 1,
\end{equation}
which has an explicit formula for ${\bm  x}$.
And the solution in \eqref{L1L2-ProximalMap4} is unique  in some special cases. Therefore the problem \eqref{ISTA1} is just as follows
\begin{align}\label{ISTA1.2}
\bm{z}^{k+1}
=\text{Proj}_{\text{Range}(\bm{D}^{\top})}\left(\text{Prox}_{\lambda(\ell_1-\alpha\ell_2)}\left(\left(\bm{z}^{k}-\gamma \nabla h(\bm{z}^{k})\right)\right)\right),
\end{align}
where $\text{Proj}_{\mathcal{C}}(\cdot)$ is a projection operator on the set $\mathcal{C}$.

So far, the  original analysis model
\eqref{VectorL1-alphaL2-ASSO} has been converted into a much simpler form \eqref{ISTA1}.
However, it is a challenge to find an analytical solution of \eqref{ISTA1} since there is  the constraint ${\bm  z}\in\text{Range}(\bm{D}^{\top})$.
Note that the orthogonal projection operator on $\text{Range}(\bm{D}^{\top})=\{\bm{\Phi}\bm{z}:\bm{z}\in\mathbb{R}^{d}\}$ is
$$
\text{Proj}_{\text{Range}(\bm{D}^{\top})}(\bm{z})=\bm{D}^{\top}\bm{\Phi}\bm{z}.
$$
Therefore, we propose to replace \eqref{ISTA1} by
\begin{align}\label{ISTA2}
\begin{cases}
\tilde{\bm{z}}^{k+1}=\text{Prox}_{\lambda\gamma(\ell_1-\alpha\ell_2)}\left(\bm{z}^{k}-\gamma \bm{\Phi}^{\top}\bm{A}^{\top}(\bm{A}\bm{\Phi}\bm{z}^{k}-\bm{b})\right),&\\
\bm{z}^{k+1}=\text{Proj}_{\text{Range}(\bm{D}^{\top})}(\tilde{\bm{z}}^{k+1})=\bm{D}^{\top}\bm{\Phi}\tilde{\bm{z}}^{k+1}.&
\end{cases}
\end{align}
By the fact that $\bm{\Phi}\bm{D}^{\top}=\bm{I}_{n}$ and \eqref{canonicaldual}, the two steps in \eqref{ISTA2} can be recast as
\begin{align}\label{ISTA3}
\tilde{\bm{z}}^{k+1}=\text{Prox}_{\lambda\gamma(\ell_1-\alpha\ell_2)}\left(\bm{D}^{\top}\left(\bm{\Phi}\tilde{\bm{z}}^{k}-\gamma (\bm{D}\bm{D}^{\top})^{-1}\bm{A}^{\top}(\bm{A}\bm{\Phi}\tilde{\bm{z}}^{k}-\bm{b})\right)\right).
\end{align}

Now, let us turn our attention to how to get the formulation of $\bm{x}^{k+1}$.
By substituting the coefficients $\bm{x}^{k}=\bm{\Phi}\bm{z}^{k}=\bm{\Phi}\bm{D}^{\top}\bm{\Phi}\tilde{\bm{z}}^{k}=
\bm{\Phi}\tilde{\bm{z}}^{k}$ into \eqref{ISTA3}, one has
\begin{align}\label{ISTA4}
\bm{x}^{k+1}=\bm{\Phi}\text{Prox}_{\lambda\gamma(\ell_1-\alpha\ell_2)}\left(\bm{D}^{\top}\left(\bm{x}^{k}-\gamma (\bm{D}\bm{D}^{\top})^{-1}\bm{A}^{\top}(\bm{A}\bm{x}^{k}-\bm{b})\right)\right),
\end{align}
which is a solution of  the $\ell_1-\alpha\ell_2$-ASSO \eqref{VectorL1-alphaL2-ASSO}.
For a tight frame, we have  $\bm{\Phi}=\bm{D}$ and $\bm{D}\bm{D}^{\top}=\bm{I}_{n}$, then \eqref{ISTA4} reduces to
\begin{align}\label{ISTA5}
\bm{x}^{k+1}=\bm{D}\text{Prox}_{\lambda\gamma(\ell_1-\alpha\ell_2)}\left(\bm{D}^{\top}\left(\bm{x}^{k}-\gamma \bm{A}^{\top}(\bm{A}\bm{x}^{k}-\bm{b})\right)\right).
\end{align}

Based on all the above derivations, the
 efficient algorithm of the $\ell_1-\alpha\ell_2$-ASSO \eqref{VectorL1-alphaL2-ASSO}
 is proposed and summarized
in Algorithm $1$ as follows.

\medskip
\noindent\rule[0.25\baselineskip]{\textwidth}{1pt}
\label{al:pFISTA}
\centerline {\bf Algorithm $1$: the $\ell_1-\alpha\ell_2$-pFISTA for solving \eqref{VectorL1-alphaL2-ASSO}}\\
{\bf Input:}\ ${\bm  A}$,${\bm  D}$, ${\bm  b}$,  $0<\alpha\leq 1$,  $\lambda$,  $\gamma$. \\
{\bf Initials:}\  $\bm{ x}=\bm{x}^0$, $\bm{ y}=\bm{y}^0=\bm{x}^0$, $t=t^0=1$, $k=0$.\\
{\bf Circulate} Step 1--Step 4 until ``some stopping criterion is satisfied":  

 ~{\bf Step 1:} Update ${\bm  x}^{k+1}$ according to
\begin{equation}\label{FISTA1}
\bm{x}^{k+1}=\bm{D}~\text{Prox}_{\lambda\gamma(\ell_1-\alpha\ell_2)}\left(\bm{D}^{\top}\left(\bm{y}^{k}-\gamma \bm{A}^{\top}(\bm{A}\bm{y}^{k}-\bm{b})\right)\right).
\end{equation}

~{\bf Step 2:} Update ${\bm  t}^{k+1}$ as follows
\begin{equation}\label{FISTA2}
t_{k+1}=\frac{1+\sqrt{1+4t_{k}^2}}{2}.
\end{equation}

~{\bf Step 3:} Update $\bm{y}^{k+1}$ as follows
\begin{equation}\label{FISTA3}
\bm{y}_{k+1}=\bm{x}^{k+1}+\frac{t_k-1}{t_{k+1}}(\bm{x}^{k+1}-\bm{x}^{k}).
\end{equation}

~{\bf Step 4:} Update $k$ to $k+1$.\\
{\bf Output:}  $\hat{\bm {x}}$.\\
\noindent\rule[0.25\baselineskip]{\textwidth}{1pt}

\begin{remark}
In our algorithm, we set the total iterated number $K=1000$, and take the stopping criterion $\|\bm{x}^{k+1}-\bm{x}^{k}\|_{2}/\|\bm{x}^{k}\|_{2}<\epsilon$ with the tolerate error $\epsilon=10^{-6}$.
\end{remark}
\section{Numerical Experiments}\label{s5}
\noindent

In this section,
we  demonstrate the performance of the $\ell_1-\alpha\ell_2$-ASSO \eqref{VectorL1-alphaL2-ASSO} via simulation experiments and compare the proposed  $\ell_1-\alpha\ell_2$-ASSO \eqref{VectorL1-alphaL2-ASSO} to
the state-of-art the $\ell_1$-analysis and $\ell_p$-analysis minimization methods.

All experiments were performed under Windows
Vista Premium and MATLAB v7.8 (R2016b) running on a Huawei laptop with an Intel(R)
Core(TM)i5-8250U CPU at 1.8 GHz and 8195MB RAM of memory.

\subsection{Signal Reconstruct Under Tight Frame}\label{s5.1}
\noindent

In this subsection, we evaluate the performance of  the $\ell_1-\alpha\ell_2$-ASSO \eqref{VectorL1-alphaL2-ASSO} and compared our method with the following models:
\begin{equation}\label{VectorL1-ABP}
\min_{{\bm  x}\in\mathbb{R}^n}~\lambda\|{\bm D}^{\top}{\bm x}\|_{p}^{p}~\text{subject~to~}{\bm  A}{\bm  x}={\bm  b}.
\end{equation}
When $p=1$, the method \eqref{VectorL1-ABP} is Analysis Basis Pursuit, which is solved by CVX package (see \cite{genzel2021analysis,nam2013cosparse}).   When  $0<p<1$,  Lin and Li \cite{lin2016restricted} present  an algorithm based on iteratively reweighted least
squares (IRLS) to solve the $\ell_{p}$-analysis model \eqref{VectorL1-ABP}.
Many papers have showed that IRLS method with smaller value of $p$ (for example $p= 0.1, 0.5$) perform better than that of larger value of $p$ (for example $p= 0.7, 0.9$). In addition, $p=0.5$ gave slightly higher success frequency than $p=0.1$. Please refer to \cite[Section 4]{chartrand2008restricted}, \cite[Section 8.1]{daubechies2010iteratively}, and \cite[Section 4.1]{lai2013improved}. Therefore,
we only compare our method with the $\ell_{p}$ model  \eqref{VectorL1-ABP}  for $p=0.5$.

First of all, we roughly follow a construction of tight random frames from \cite{nam2013cosparse}:
\begin{enumerate}
\item[(i)]First, draw a $n\times d$ Gaussian random matrix $\bm{E}$ and compute its singular value decomposition $\bm{E}=\bm{U}\bm{\Sigma}\bm{V}^{\top}$.
\item[(ii)] If $n\leq d$, we replace $\bm{\Sigma}$ by the matrix $\tilde{\bm{\Sigma}}=[\tau\bm{I}_n, \bm{0}]\in\mathbb{R}^{n\times d}$ with $\tau=\sqrt{d/n}$, which yields a tight frame $\bm{D}=\bm{U}\tilde{\bm{\Sigma}}\bm{V}^{\top}$. If $n> d$, we replace $\bm{\Sigma}$ by the matrix $\tilde{\bm{\Sigma}}=[\tau\bm{I}_d, \bm{0}]^{\top}\in\mathbb{R}^{n\times d}$ with $\tau=\sqrt{d/n}$, which yields a tight frame $\bm{D}=\bm{U}\tilde{\bm{\Sigma}}\bm{V}^{\top}$.\\
\end{enumerate}

Nam et.al. \cite{nam2013cosparse} also showed us how to generate cosparse signal $\bm{x}_{0}\in\mathbb{R}^{n}$.
We adopt their scheme and produce an $s$-cosparse signal in the following way:
\begin{enumerate}
\item[(a)] First, choose $s$ rows of the analysis operator $\bm{D}^{\top}=\bm{\Omega}\in\mathbb{R}^{n\times d}$ at random, and those are
denoted by an index set $|S|$ (thus, $|S|=s$).
\item[(b)] Second, form an arbitrary signal $\bm{y}$ in $\mathbb{R}^n$--e.g., a random vector with Gaussian i.i.d. entries.
\item[(c)] Then, project $\bm{y}$ onto the orthogonal complement of the subspace generated by the rows of $\bm{\Omega}$ that are indexed by $S$, this way getting the cosparse signal $\bm{x}_0$. Explicitly, $\bm{x}_0=\left(\bm{I}_n-\bm{\Omega}_{S}^{\top }(\bm{\Omega}_{S}\bm{\Omega}_{S}^{\top})^{-1}\bm{\Omega}_{S}\right)\bm{y}$. In fact, $\bm{D}^{\top}\bm{x}_0=[\bm{0};\bm{\Omega}_{S^c}\bm{x}_0]\in\mathbb{R}^{d}$ is $(d-s)$-sparse.
\end{enumerate}
Alternatively, one could first find a
basis for the orthogonal complement and then generate a random coefficient vector for the basis.
In the experiment, the entries of $\bm{A}\in\mathbb{R}^{m\times n}$ are drawn independently from the normal distribution. The observation is obtained by $\bm{b} =\bm{A}\bm{x}_{0}$.

Let $\hat{\bm x}$ be the reconstructed signal.
We record the success rate over $100$  independent trials. The recovery is regarded as successful if
\begin{equation}\label{rel.err}
\text{rel-err}(\hat{\bm  x},{\bm  x}_0)=\frac{\|\hat{\bm  x}-{\bm  x}_0\|_2}{\|{\bm  x}_0\|_2}<\varepsilon,
\end{equation}
for $\varepsilon=10^{-2}$.
We display success rate of different algorithms to recover sparse signals over $100$ repeated trials for  different cosparsity $s$.

For fairness of comparison, the key parameters of our proposed method and compared algorithms have been tuned in all experiments according to \cite{nam2013cosparse}. In all cases, the signal dimension $n$ is set to 100. We then varied the number $m$ of measurements, the cosparsity $\ell$ of the target signal, and the operator size $d$ according to the following formulae:
\begin{equation}\label{cosparsity.setup}
m=\varrho n, d=\varsigma n,  \ell=n-\rho m
\end{equation}
where $0<\varrho \leq 1$, $\varsigma\geq 1$, $0<\rho\leq1$. Here we take  $\varsigma=1, \rho=\{0.05,0.10,0.15,\ldots,1\}$ and $\varrho=\{0.05,0.1,0.15,\ldots,1\}$, i.e., the measurement $m=\{5,10,15,\dots,100\}$.

In Figure \ref{figure.SuccnumberPhaseTrransition-Comparion}, we plot the phase transition diagram, which characterizes sharp shifts in the success probability of reconstruction when the dimension parameter crosses a threshold. The $x$-axis and the $y$-axis represent the under-sampling ratio and co-sparsity ratio, respectively. Yellow and blue denote perfect recovery and
failure in all experiments, respectively. It can be clearly seen that Success Rate (the yellow) of the proposed $\ell_1-\alpha\ell_2$ methods are the highest in all experiments. Experimental results show that $\ell_1-\alpha\ell_2$ methods outperform $\ell_1$ method and $\ell_p$ methods.

Figure \ref{figure.CPUtimePhaseTrransition-Comparion} shows the average CPU time of all methods at different $\varrho$ and $\rho$. We can observe that the CPU time of the proposed method is significantly lower than those of $\ell_{1}$ method at whole, and  higher than those of $\ell_p$ method for $p=0.5$. Thus, $\ell_1-\alpha\ell_2$-ASSO method can achieve a good balance between CPU time and recovery performance.

\begin{figure*}[t] 
\begin{tabular}{ccc}
\includegraphics[width=4.5cm,height=4.5cm]{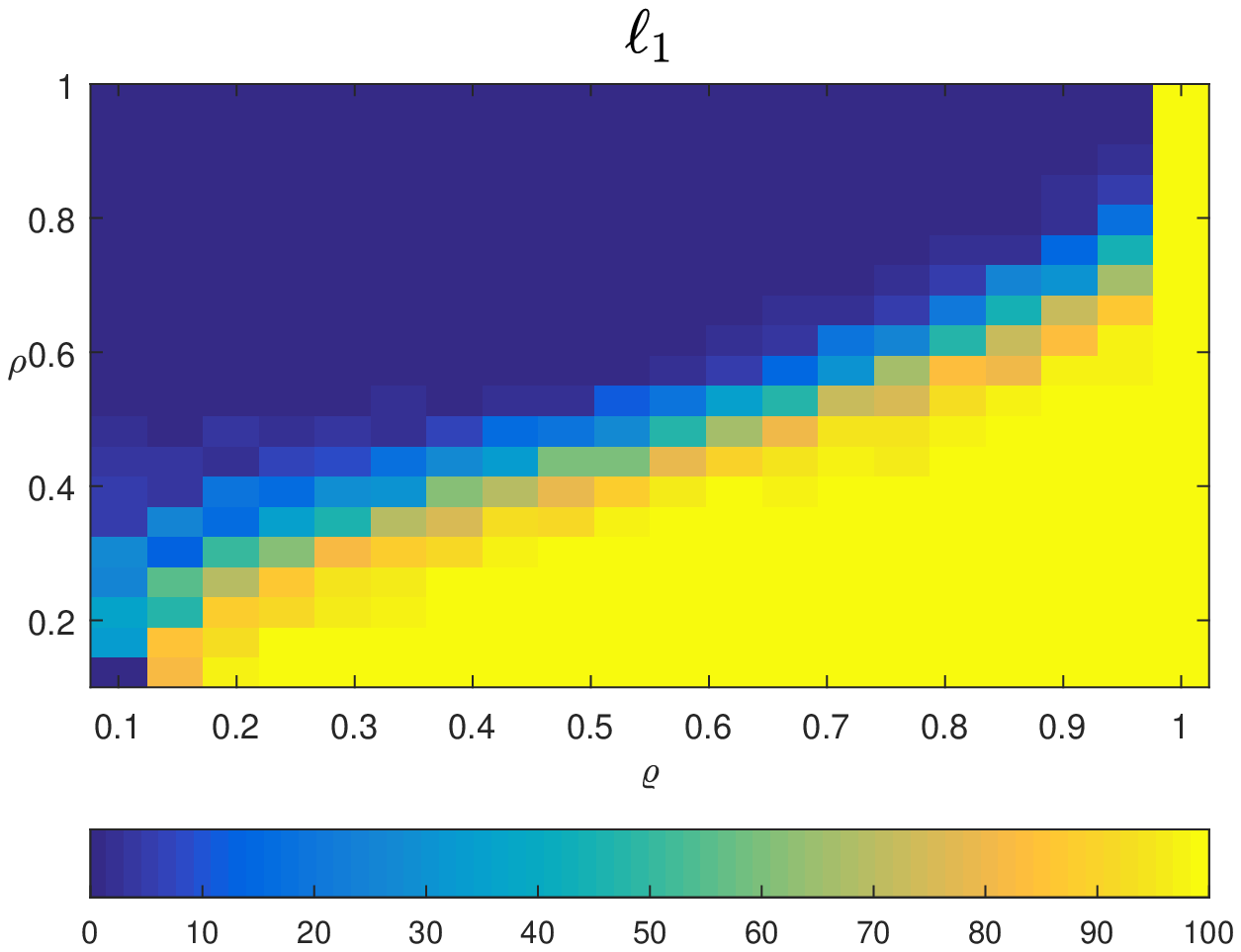}&
\includegraphics[width=4.5cm,height=4.5cm]{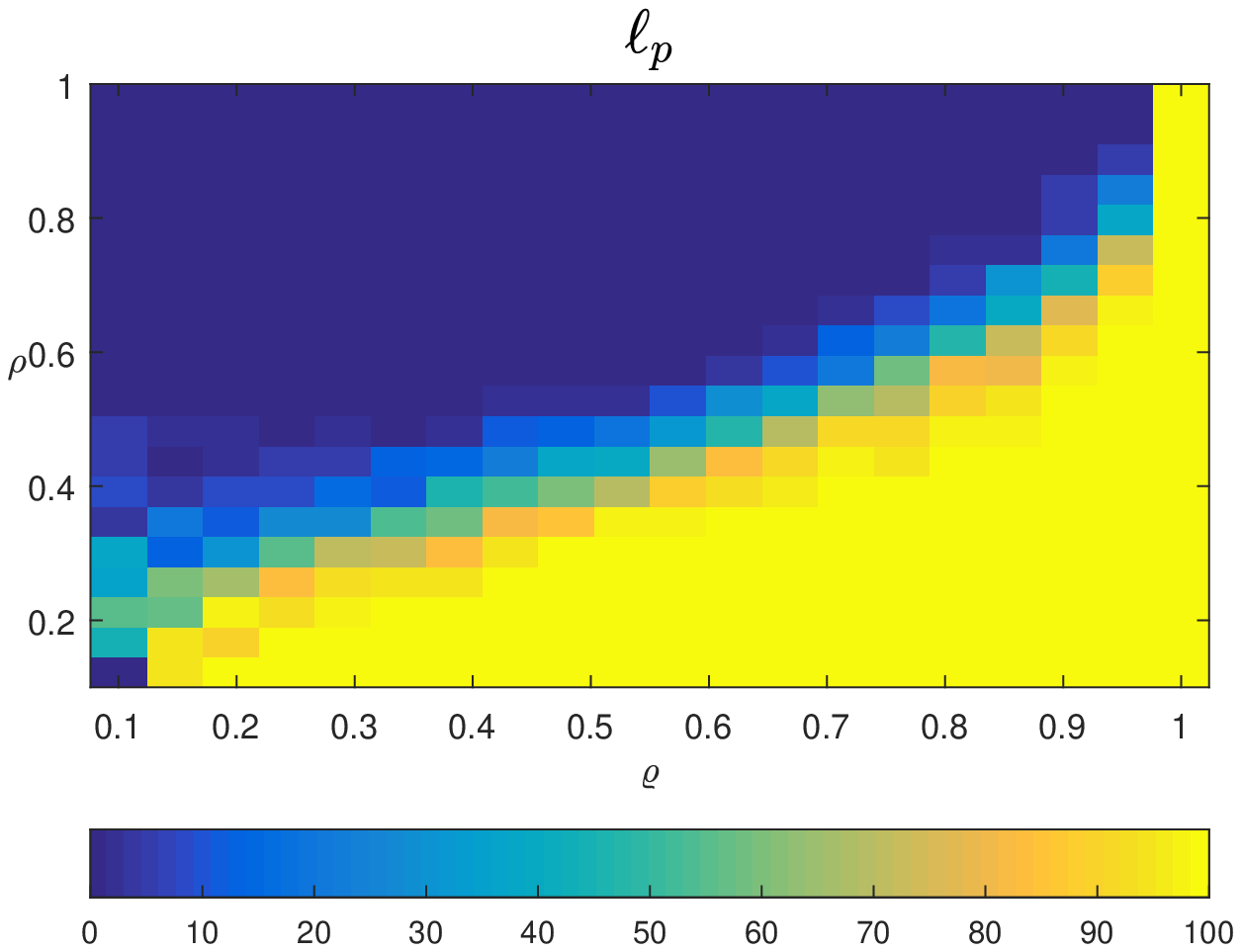}&
\includegraphics[width=4.5cm,height=4.5cm]{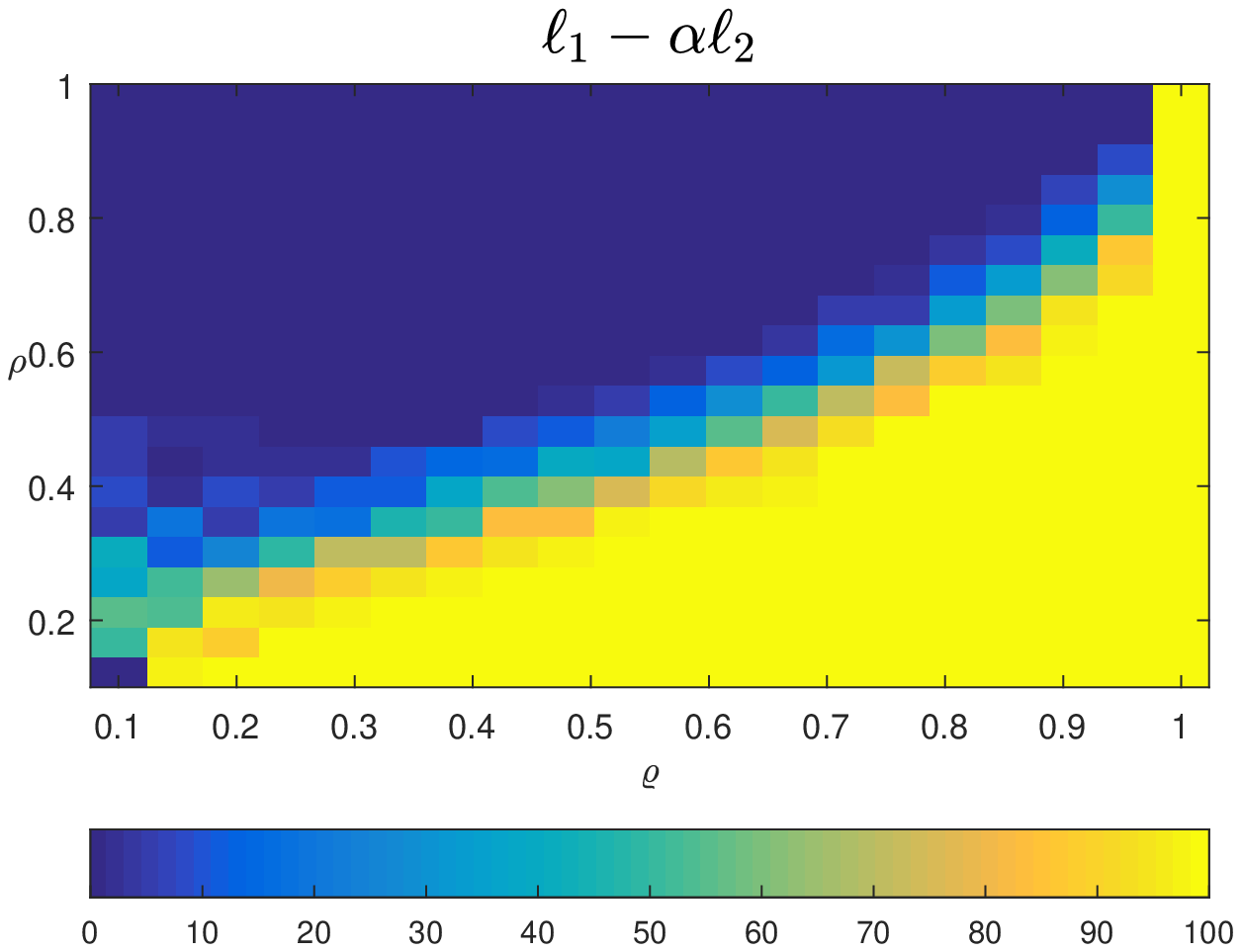}
\end{tabular}
\centering
\caption{Success  percentage of the  $\ell_1$-, $\ell_p$-~($p=0.5)$ and $\ell_{1}-\alpha\ell_{2}$-analysis for recover sparse signals
versus the ratios $\varrho$ and $\rho$.
}
\label{figure.SuccnumberPhaseTrransition-Comparion}
\end{figure*}

\begin{figure*}[t] 
\begin{tabular}{c}
\includegraphics[width=16.0cm]{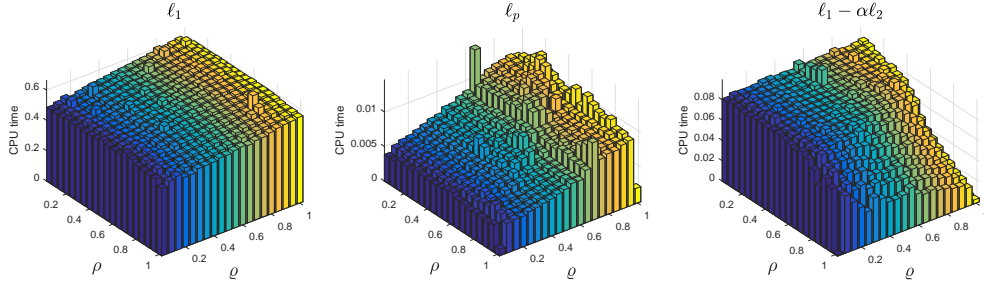}
\end{tabular}
\centering
\caption{  CPU time of the  $\ell_1$-, $\ell_p$-~($p=0.5)$ and $\ell_{1}-\alpha\ell_{2}$-analysis  for the sparse signal recovery
versus the ratios $\varrho$ and $\rho$.}
\label{figure.CPUtimePhaseTrransition-Comparion}
\end{figure*}

\subsection{Reconstruction of Compressed Sensing Magnetic Resonance Imaging under Tight Frame}\label{s5.2}
\noindent

 In this subsection, we consider the  shift-invariant discrete wavelet transform (SIDWT) for tight frame $\bm{D}$, which
is  a typical tight frame in simulation \cite{baker2011translational,baraniukrice,coifman1995translation,kayvanrad2014stationary}. And SIDWT is also called as undecimated, translation-invariant, or
fully redundant wavelets. In all the experiments,  we utilize Daubechies wavelets with 4 decomposition levels in SIDWT.

In CS-MRI, the sampling operator  is
$$\bm{A}=\bm{U}\mathcal{F},$$
where $\mathcal{F}$  is  the discrete Fourier transform,  and $\bm{U}$ is the sampling mask  in the frequency
space. The matrix $\bm{U}$ is also called the undersampling matrix. We keep samples along certain radial lines passing through the center of the Fourier data ($k$-space).
We reconstruct magnetic resonance images (MRI) from  incomplete spectral Fourier data: $256\times 256$ Brain MRI and $512\times 512$ Foot MRI (see  \cite[Section 4.1]{li2020compressive}) via the $\ell_1-\alpha\ell_2$-ASSO \eqref{VectorL1-alphaL2-ASSO}.

Similarly, we compare the $\ell_1-\alpha\ell_2$-ASSO \eqref{VectorL1-alphaL2-ASSO} to the $\ell_1$-analysis and $\ell_p$($0<p<1$)-analysis minimization methods.
 We adopt the pFISTA for tight frames  in  \cite{liu2016projected}  to solve the  $\ell_1$-analysis minimization problem.
The $\ell_p$($0<p<1$)-analysis model is
\begin{equation}\label{Lp-Analysis}
\min_{{\bm  x}}~\lambda\|{\bm D}^{\top}{\bm x}\|_{p}^{p}+\frac{1}{2}\|\bm{UF}{\bm  x}-{\bm  b}\|_{2}^{2},
\end{equation}
which is solved by the idea of the pFISTA for tight frames.  In fact,  as shown in section \ref{s4},
the equality  \eqref{FISTA1} is replaced by
\begin{equation}\label{FISTALp}
\bm{x}^{k+1}=\bm{D}~\text{Prox}_{\lambda\gamma\ell_p}\left(\bm{D}^{\top}\left(\bm{y}^{k}-\gamma \bm{A}^{*}(\bm{A}\bm{y}^{k}-\bm{b})\right)\right),
\end{equation}
where $0<p<1$ and the notation ${Prox}_{\lambda\ell_p}(\bf{b})$ is the proximal operator of $\ell_{p}$ norm, see \cite{marjanovic2012optimization}.

The quantitative comparison is done in terms of the relative error (RE) defined as
$$
\text{RE}=\frac{\|\hat{\bm{x}}-\bm{x}_0\|_{2}}{\|\bm{x}_0\|_{2}},
$$
where $\bm{x}_0$ is the truth image and $\hat{\bm{x}}$ is the reconstructed
image. To demonstrate how  $\ell_1-\alpha\ell_2$-ASSO method compares with other methods in terms of image quality,
 we show the restored versions of  Brain images and Foot images and reconstruction errors in Figures
\ref{figure.Reconstruct-image-comparsion}, \ref{figure.Reconstruct-FootMRI-comparsion} and Table \ref{tab:MRI-Time}, respectively.

In Figures \ref{figure.Reconstruct-image-comparsion} and \ref{figure.Reconstruct-FootMRI-comparsion},
 we show the reconstructed images of different methods for $76$ radial sampling lines
 (sampling rate 30.81$\%$ and 16.17$\%$ for Brain-MRI and Foot-MRI, respectively). By inspecting the recovered images of brain, it can be seen that $\ell_1-\alpha\ell_2$ method can obtain better performance than other methods.
We also record the CPU time of all methods in Table \ref{tab:MRI-Time}.

\section{Conclusions }\label{s6}
\noindent

In this paper, we consider the signal   and compressed sensing magnetic resonance imaging reconstruction under tight frame. We propose the unconstrained $\ell_{1}-\alpha\ell_{2}$-analysis model
\eqref{VectorL1-alphaL2-ASSO} and \eqref{VectorL1-alphaL2-RASSO}. Based on the restricted isometry
property and restricted orthogonality constant adapted to tight frame $\bm{D}$ ($\bm{D}$-RIP and $\bm{D}$-ROC), we develop new vital auxiliary tools (see Propositions \ref{prop.DROC} and \ref{NonsparseROC}) and sufficient conditions of stable recovery (see Theorems \ref{StableRecoveryviaVectorL1-alphaL2-ASSO} and \ref{StableRecoveryviaVectorL1-alphaL2-RASSO}). Based on the Projected FISTA \cite{liu2016projected},
 we  establish the  fast and  efficient algorithm
 to solve the unconstrained $\ell_{1}-\alpha\ell_{2}$-analysis model
    in Section \ref{s4}.
The proposed  method has better performance than the  $\ell_p$-analysis model with $0<p\leq 1$
in  numerical examples for  the  signal and compresses sensing MRI recovery.

\begin{figure*}[htbp!]
\setlength{\tabcolsep}{4.0pt}\small
\begin{tabular}{c}
\includegraphics[width=16.0cm,height=8.0cm]{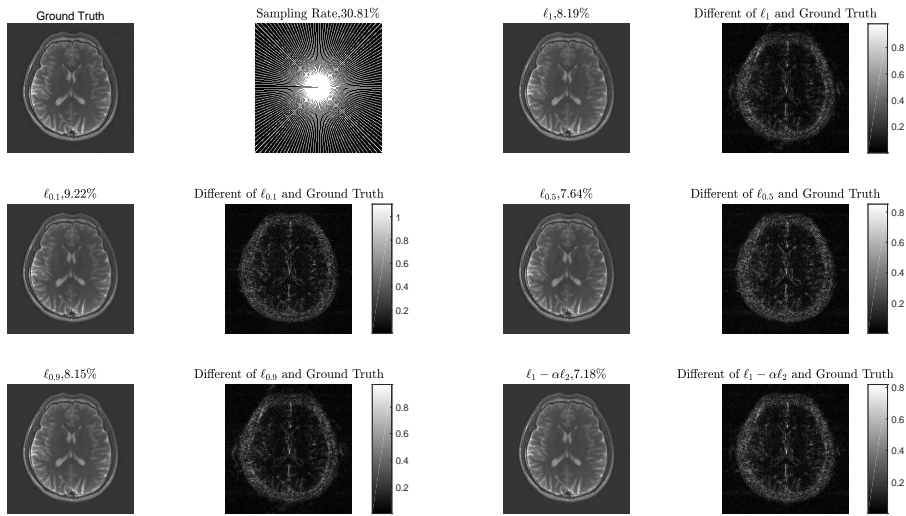}
\end{tabular}
\centering
\caption{\label{figure.Reconstruct-image-comparsion} Reconstructed Brain-MRI by the $\ell_1$-, $\ell_p$-~($0<p<1$) and $\ell_{1}-\alpha\ell_{2}$-analysis. From left to right in the first line: Ground truth,  sample lines, $\ell_1$ reconstruction image, difference images of $\ell_1$ to the ground truth image. From left to right in the second line: $\ell_{0.1}$ reconstruction image, difference images of $\ell_{0.1}$ to the ground truth image, $\ell_{0.5}$ reconstruction image, difference images of $\ell_{0.5}$ to the ground truth image. From left to right in the third line: $\ell_{0.9}$ reconstruction image, difference images of $\ell_{0.9}$ to the ground truth image, $\ell_1-\alpha\ell_2$ reconstruction image, difference images of $\ell_1-\alpha\ell_2$ to the ground truth image.
}
\vspace{-0.1cm}
\end{figure*}

\begin{figure*}[htbp!]
\setlength{\tabcolsep}{4.0pt}\small
\begin{tabular}{c}
\includegraphics[width=16.0cm,height=8.0cm]{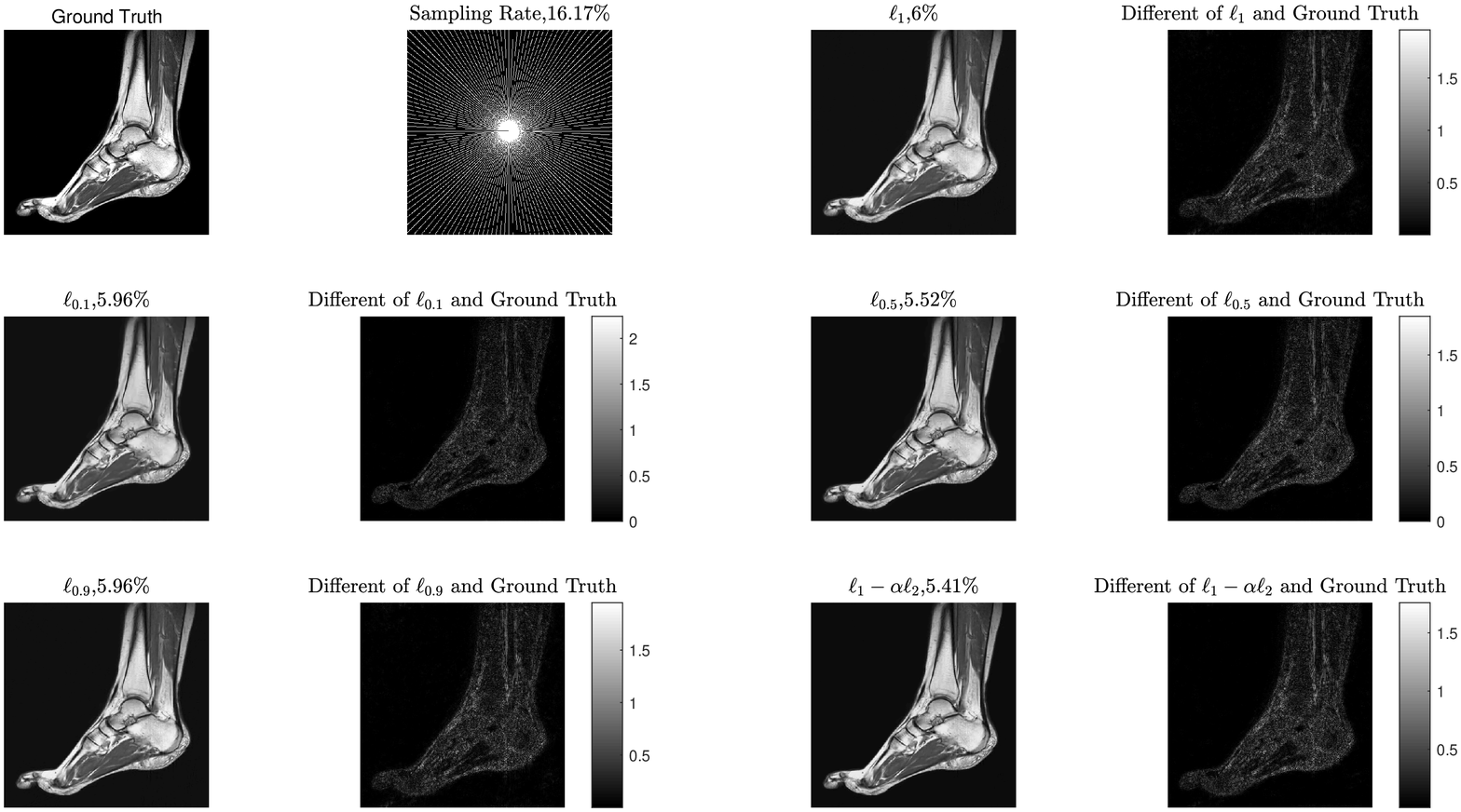}
\end{tabular}
\centering
\caption{\label{figure.Reconstruct-FootMRI-comparsion} Reconstructed Foot-MRI by the $\ell_1$-, $\ell_p$-~($0<p<1$) and $\ell_{1}-\alpha\ell_{2}$-analysis. From left to right in the first line: Ground truth,  sample lines, $\ell_1$ reconstruction image, difference images of $\ell_1$ to the ground truth image. From left to right in the second line: $\ell_{0.1}$ reconstruction image, difference images of $\ell_{0.1}$ to the ground truth image, $\ell_{0.5}$ reconstruction image, difference images of $\ell_{0.5}$ to the ground truth image. From left to right in the third line: $\ell_{0.9}$ reconstruction image, difference images of $\ell_{0.9}$ to the ground truth image, $\ell_1-\alpha\ell_2$ reconstruction image, difference images of $\ell_1-\alpha\ell_2$ to the ground truth image.
}
\vspace{-0.1cm}
\end{figure*}

\begin{table}[htbp]
\setlength{\tabcolsep}{5pt}\small
\begin{center}
\caption{The CPU Time (s) of Different reconstruction Models}\label{tab:MRI-Time}
\begin{tabular}{|c|c|c|c|c|c|c|}\hline
Image    & Sampling Rate &$\ell_{1}$  &$\ell_{0.1}$ &$\ell_{0.5}$ &$\ell_{0.9}$ &$\ell_{1}-\alpha\ell_{2}$  \\\hline
Brain-MRI    &30.08$\%$  &80.9519      &164.4825      &238.4073   &322.2761   &85.1567  \\\hline
Foot-MRI     &16.17$\%$ &308.3546         &441.8828      &657.7717  &857.9363   &340.5498 \\\hline
\end{tabular}
\end{center}
\end{table}

\newpage
\begin{appendices}
\section{The proof of Lemma \ref{lem:CrossItem}}\label{appendx1}
\noindent

\begin{proof}
$(\bm i)$ From  the condition \eqref{lem:CrossItem.eq1}, 
it follows that
\begin{align}\label{Coneconstraintinequality}
\|\bm{D}_{S^c}^{\top}\bm{h}\|_1-\|\bm{D}_{S^c}^{\top}\bm{h}\|_2
&\overset{(a)}{\leq}
\|\bm{D}_{S^c}^{\top}\bm{h}\|_1-\alpha\|\bm{D}_{S^c}^{\top}\bm{h}\|_2\nonumber\\
&\leq a\|\bm{D}_{S}^{\top}\bm{h}\|_1+b\|\bm{D}_{S}^{\top}\bm{h}\|_2+c\|\bm{D}_{T^c}^{\top}\bm{x}\|_1+\eta\|{\bm Ah}\|_2
+\gamma\nonumber\\
&\overset{(b)}{\leq}(a\sqrt{s}+b)\|\bm{D}_{S}^{\top}\bm{h}\|_2+c\|\bm{D}_{T^c}^{\top}\bm{x}\|_1+\eta\|{\bm Ah}\|_2
+\gamma\nonumber\\
&\leq(s-\sqrt{s})\bigg(
\frac{a\sqrt{s}+b}{\sqrt{s}-1}\frac{\|\bm{D}_{S}^{\top}\bm{h}\|_2}{\sqrt{s}}+\frac{c\|\bm{D}_{T^c}^{\top}\bm{x}\|_1+\eta\|{\bm Ah}\|_2
+\gamma}{s-\sqrt{s}}\bigg)\nonumber\\
&=:(s-\sqrt{s})\varrho,
\end{align}
where  $(a)$ is due to $0<\alpha\leq 1$, and
$(b)$ follows from the fact $\|\bm{D}_{S}^{\top}\bm{h}\|_1\leq \sqrt{s}\|\bm{D}_{S}^{\top}\bm{h}\|_2$.

 Furthermore, using the fact that $(a-1)\sqrt{s}+(b+1)\geq0$, i.e.,
$\frac{a\sqrt{s}+b}{\sqrt{s}-1}\geq1$,
one has
\begin{align}\label{e:etainftynew}
\|\bm{D}_{S^c}^{\top}\bm{h}\|_{\infty}\leq \frac{\|\bm{D}_{S}^{\top}\bm{h}\|_{1}}{s}\leq\frac{\|\bm{D}_{S}^{\top}\bm{h}\|_{2}}{\sqrt{s}}\leq\frac{\sqrt{s}+\alpha}{\sqrt{s}-1}\frac{\|\bm{D}_{S}^{\top}\bm{h}\|_2}{\sqrt{s}}\leq \varrho,
\end{align}
where the last inequality is due to the definition of $\varrho$ in \eqref{Coneconstraintinequality}.
By Proposition \ref{NonsparseROC} with $\bm{u}=\|\bm{D}_{S}^{\top}\bm{h}\|_1$
and $\bm v=\|\bm{D}_{S^c}^{\top}\bm{h}\|_1$,  the desired inequality  \eqref{e:CrossItem1} is clear.

$(\bm{ii})$
From the definition of $\tilde{S}$ in \eqref{def:S} and $\varrho$ in \eqref{Coneconstraintinequality}, it follows that
\begin{align}\label{e:etaW2infty}
\|\bm{D}_{{\tilde{S}}^c}^{\top}\bm{h}\|_{\infty}\leq\frac{\varrho}{t-1},
\end{align}
and
\begin{align}\label{e:etaW2}
\|\bm{D}_{{\tilde{S}}^c}^{\top}\bm{h}\|_1-\|\bm{D}_{{\tilde{S}}^c}^{\top}\bm{h}\|_2
=&\|\bm{D}_{S^c}^{\top}\bm{h}-\bm{D}_{\tilde{S}\setminus S}^{\top}\bm{h}\|_{1}
-\|\bm{D}_{S^c}^{\top}\bm{h}-\bm{D}_{\tilde{S}\setminus S}^{\top}\bm{h}\|_{2}\nonumber\\
\overset{(a)}{=}&\|\bm{D}_{S^c}^{\top}\bm{h}\|_1-\|\bm{D}_{\tilde{S}\setminus S}^{\top}\bm{h}\|_{1}
-\|\bm{D}_{S^c}^{\top}\bm{h}-\bm{D}_{\tilde{S}\setminus S}^{\top}\bm{h}\|_{2}\nonumber\\
\overset{(b)}{\leq}&\big(\|\bm{D}_{S^c}^{\top}\bm{h}\|_1-\|\bm{D}_{S^c}^{\top}\bm{h}\|_2\big)
-(\|\bm{D}_{\tilde{S}\setminus S}^{\top}\bm{h}\|_{1}-\|\bm{D}_{\tilde{S}\setminus S}^{\top}\bm{h}\|_{2})\nonumber\\
\overset{(c)}{\leq}&(s-\sqrt{s})\varrho-(\|\bm{D}_{\tilde{S}\setminus S}^{\top}\bm{h}\|_{1}-\|\bm{D}_{\tilde{S}\setminus S}^{\top}\bm{h}\|_{2}),
\end{align}
where $(a)$, $(b)$ and $(c)$  follow from  $\tilde{S}\setminus S \subseteq S^{c}$,  the triangle inequality on $\|\cdot\|_2$, and \eqref{Coneconstraintinequality}, respectively.

For the second term of the above inequality,  using  Lemma \ref{LocalEstimateL1-L2} (b) with $S_1=\tilde{S}\setminus S$ and $S_2=\tilde{S}^c$, we derive that
\begin{align}\label{e:eta-maxkuplowbound}
\|\bm{D}_{S^c}^{\top}\bm{h}\|_1-\|\bm{D}_{S^c}^{\top}\bm{h}\|_2
\geq&\big(\|\bm{D}_{\tilde{S}\setminus S}^{\top}\bm{h}\|_1-\|\bm{D}_{\tilde{S}\setminus S}^{\top}\bm{h}\|_2\big)
+\big(\|\bm{D}_{\tilde{S}^c}^{\top}\bm{h}\|_1-\|\bm{D}_{\tilde{S}^c}^{\top}\bm{h}\|_2\big)\nonumber\\
\geq&\|\bm{D}_{\tilde{S}\setminus S}^{\top}\bm{h}\|_1-\|\bm{D}_{\tilde{S}\setminus S}^{\top}\bm{h}\|_2\nonumber\\
\overset{(a)}{\geq}&(|\tilde{S}\setminus S|-\sqrt{|\tilde{S}\setminus S|})\min_{i\in \tilde{S}\setminus S}|(\bm{D}_{\tilde{S}\setminus S}^{\top}\bm{h})(i)|\nonumber\\
\overset{(b)}{\geq}&(|\tilde{S}\setminus S|-\sqrt{|\tilde{S}\setminus S|})\frac{\varrho}{t-1},
\end{align}
where we use Lemma \ref{LocalEstimateL1-L2} (a) and  the definition of $\tilde{S}$ in $(a)$ and $(b)$, respectively.
Substituting   \eqref{e:eta-maxkuplowbound} into  \eqref{e:etaW2},  there is
\begin{align}\label{e:l1-2upperbounds}
\|\bm{D}_{\tilde{S}^c}^{\top}\bm{h}\|_1-\|\bm{D}_{\tilde{S}^c}^{\top}\bm{h}\|_2
\leq \Big((s(t-1)-|\tilde{S}\setminus S|)-
(\sqrt{s}(t-1)-\sqrt{|\tilde{S}\setminus S|})\Big)\frac{\varrho}{t-1}.
\end{align}
Since $t\geq3$ and $s\geq 2$,
 as shown in the items (a) and (b) of
\cite[Page 18]{ge2021dantzig}, we have
$$|\tilde{S}\setminus S|< s(t-1),\ \ \ \ \sqrt{s(t-1)-|\tilde{S}\setminus S|}\leq \sqrt{s}(t-1)-\sqrt{|\tilde{S}\setminus S|}.$$
Then,
\begin{align}\label{e:etaW2upperbound}
\|\bm{D}_{\tilde{S}^c}^{\top}\bm{h}\|_1-\|\bm{D}_{\tilde{S}^c}^{\top}\bm{h}\|_2
\leq
\Big(s(t-1)-|\tilde{S}\setminus S|-\sqrt{s(t-1)-|\tilde{S}\setminus S|}\Big)\frac{\varrho}{t-1}.
\end{align}

Therefore,
 from  \eqref{e:etaW2infty}, \eqref{e:etaW2upperbound} and Proposition \ref{NonsparseROC} with  $\bm u=\bm{D}_{\tilde{S}}^{\top}{\bm h}$,
$\bm {v}=\bm{D}_{\tilde{S}^c}^{\top}\bm{h}$,
it follows that
\begin{align}\label{e:omegal2high}
&|\langle \bm{AD}\bm{D}_{\tilde{S}}^{\top}{\bm h},\bm{AD}\bm{D}_{\tilde{S}^c}^{\top}{\bm h}\rangle
+\langle\bar{\bm{ D}}\bm{D}_{\tilde{S}}^{\top}{\bm h},\bar{\bm{ D}}\bm{D}_{\tilde{S}^c}^{\top}{\bm h} \rangle|\nonumber\\
&\leq\bigg(1+\frac{\sqrt{2}}{2}\bigg)\theta_{ts, (t-1)s-|\tilde{S}\setminus S|}\sqrt{\lceil(t-1)s\rceil-|\tilde{S}\setminus S|} \frac{\varrho}{t-1}\|\bm{D}_{\tilde{S}}^{\top}{\bm h}\|_2\nonumber\\
&\leq \bigg(1+\frac{\sqrt{2}}{2}\bigg)\theta_{ts, (t-1)s}\sqrt{\lceil(t-1)s\rceil} \frac{\varrho}{t-1}\|\bm{D}_{\tilde{S}}^{\top}{\bm h}\|_2.
\end{align}
Based on the fact $\bar{\bm{ D}}\bm{D}^{\top}=\bm{0}$, the above inequality reduces to  the desired \eqref{e:CrossItem2}.

$(\bm{iii})$
For the term 
$\langle \bm{Ah}, \bm{AD}\bm{D}_{\tilde{S}}^{\top}\bm{h}\rangle$, there is
\begin{align}\label{Upperbound}
&|\langle \bm{Ah}, \bm{AD}\bm{D}_{\tilde{S}}^{\top}\bm{h}\rangle\big|\leq\| \bm{Ah}\|_{2}\|\bm{AD}\bm{D}_{\tilde{S}}^{\top}\bm{h}\|_{2}
\overset{(a)}{\leq}\sqrt{1+\delta_{ts}}\|\bm{D}\bm{D}_{\tilde{S}}^{\top}\bm{h}\|_{2}\|\bm{Ah}\|_{2}\nonumber\\
&\leq\sqrt{1+\delta_{ts}}\|\bm{D}\|_{2\rightarrow 2}^{1/2}\|\bm{D}_{\tilde{S}}^{\top}\bm{h}\|_{2}\|\bm{Ah}\|_{2}
\overset{(b)}{=} \sqrt{1+\delta_{ts}}\|\bm{D}_{\tilde{S}}^{\top}\bm{h}\|_{2}\|\bm{Ah}\|_{2},
\end{align}
where $(a)$ is because of the matrix $\bm{A}$ satisfying the $\bm D$-RIP of $ts$ order,
 and $(b)$ follows from $\|\bm{D}\|_{2\rightarrow 2}=\|\bm{D}\bm{D}^{\top}\|_{2\rightarrow 2}^{1/2}=1$.

Next, we will establish the lower bound of $\big|\langle \bm{Ah}, \bm{AD}\bm{D}_{\tilde{S}}^{\top}\bm{h}\rangle\big|$. Note that
\begin{align*}
\big|\langle \bm{Ah}, \bm{AD}\bm{D}_{\tilde{S}}^{\top}\bm{h}\rangle\big|\geq
\|  \bm{AD}\bm{D}_{\tilde{S}}^{\top}\bm{h}\|_2^2-\big|\langle  \bm{AD}\bm{D}_{\tilde{S}^c}^{\top}\bm{h}, \bm{AD}\bm{D}_{\tilde{S}}^{\top}\bm{h}\rangle\big|.
\end{align*}
From Proposition \ref{prop.DROC} with $\bm v=\bm{D}_{\tilde{S}^c}^{\top}\bm{h}$ and $\bar{\bm{ D}}\bm{D}^{\top}=\bm{0}$, it follows  that
\begin{align*}
\| \bm{AD}\bm{D}_{\tilde{S}}^{\top}\bm{h}\|_2^2
=\| \bm{AD}\bm{D}_{\tilde{S}}^{\top}\bm{h}\|_2^2+\|\bar{\bm{D}}\bm{D}_{\tilde{S}}^{\top}\bm{h}\|_2^2
\geq(1-\delta_{ts})\|\bm{D}_{\tilde{S}}^{\top}\bm{h}\|_2^2.
\end{align*}
By \eqref{e:CrossItem2} in item (ii) and $\varrho$ in \eqref{Coneconstraintinequality}, we have
\begin{align*}
\big|\langle  \bm{AD}\bm{D}_{\tilde{S}^c}^{\top}\bm{h}, \bm{AD}\bm{D}_{\tilde{S}}^{\top}\bm{h}\rangle\big|
\leq \theta_{ts, (t-1)s}\sqrt{\lceil(t-1)s\rceil}\bigg(1+\frac{\sqrt{2}}{2}\bigg) \frac{\varrho}{t-1}\|\bm{D}_{\tilde{S}}^{\top}\bm{h}\|_2
\end{align*}
Then, 
\begin{align}\label{Lowerbound}
&\big|\langle \bm{Ah}, \bm{AD}\bm{D}_{\tilde{S}}^{\top}\bm{h}\rangle\big|\nonumber\\
&\geq(1-\delta_{ts})\|\bm{D}_{\tilde{S}}^{\top}\bm{h}\|_2^2-\theta_{ts, (t-1)s}
\bigg(1+\frac{\sqrt{2}}{2}\bigg)\frac{\sqrt{\lceil(t-1)s\rceil}}{t-1}\|\bm{D}_{\tilde{S}}^{\top}\bm{h}\|_2\nonumber\\
&\cdot\bigg(\frac{a\sqrt{s}+b}{\sqrt{s}-1}\frac{\|\bm{D}_{S}^{\top}\bm{h}\|_2}{\sqrt{s}}
+\frac{c\|\bm{D}_{T^c}^{\top}\bm{x}\|_1+\eta\|{\bm Ah}\|_2+\gamma}{s-\sqrt{s}}\bigg)\nonumber\\
&\geq\bigg( 1-\delta_{ts}-\bigg(1+\frac{\sqrt{2}}{2}\bigg)\sqrt{\frac{\lceil(t-1)s\rceil}{(t-1)^2s}}\frac{a\sqrt{s}+b}{\sqrt{s}-1}
\theta_{ts, (t-1)s }\bigg)\|\bm{D}_{\tilde{S}}^{\top}\bm{h}\|_2^2\nonumber\\
&-\theta_{ts,(t-1)s}\bigg(1+\frac{\sqrt{2}}{2}\bigg)\frac{\sqrt{\lceil(t-1)s\rceil}}{t-1}\frac{c\|\bm{D}_{T^c}^{\top }\bm{x}\|_1+\eta\|{\bm Ah}\|_2+\gamma}{s-\sqrt{s}}\|\bm{D}_{\tilde{S}}^{\top}\bm{h}\|_2,
\end{align}
where the last inequality is due to $S\subseteq \tilde{S}$.
Combining \eqref{Lowerbound} with \eqref{Upperbound}, one has
\begin{align*}
&\bigg( 1-\delta_{ts}-\sqrt{\frac{\lceil(t-1)s\rceil}{(t-1)^2s}}\frac{(\sqrt{2}+1)(\sqrt{s}+\alpha)}{\sqrt{2}(\sqrt{s}-1)}
\theta_{ts,(t-1)s}\bigg)\|\bm{D}_{\tilde{S}}^{\top}\bm{h}\|_2^2\nonumber\\
&-\bigg(\theta_{ts,(t-1)s}\frac{\sqrt{2}+1}{\sqrt{2}}\frac{\sqrt{\lceil(t-1)s\rceil}}{t-1}\frac{c\|\bm{D}_{T^c}^{\top}\bm{x}\|_1+\eta\|\bm{ Ah}\|_2+\gamma}{s-\sqrt{s}}\nonumber\\
&+\sqrt{1+\delta_{ts}}\|\bm{Ah}\|_{2}\bigg)\|\bm{D}_{\tilde{S}}^{\top}\bm{h}\|_2\leq 0.
\end{align*}
Therefore, using \eqref{RIPConditiona1add} we derive that
\begin{align}\label{Estimate.hmax(s).eq1}
&\|\bm{D}_{S}^{\top}\bm{h}\|_2\leq\|\bm{D}_{\tilde{S}}^{\top}\bm{h}\|_2\nonumber\\
&\leq\frac{\theta_{ts,(t-1)s}}{(1-\rho_{s,t})}\frac{\sqrt{2}+1}{\sqrt{2}}\frac{\sqrt{\lceil(t-1)s\rceil}}{(t-1)(s-\sqrt{s})}
\left(c\|\bm{D}_{T^c}^{\top}\bm{x}\|_1+\gamma\right)\nonumber\\
&+\bigg(\frac{\theta_{ts,(t-1)s}}{(1-\rho_{s,t})}\frac{\sqrt{2}+1}{\sqrt{2}}\frac{
\sqrt{\lceil(t-1)s\rceil}}{(t-1)(s-\sqrt{s})} +\frac{\sqrt{1+\delta_{ts}}}{1-\rho_{s,t}}\bigg)\eta\|\bm{Ah}\|_{2},
\end{align}
where $\eta\geq 1$.

$(\bm{iv})$ 
As shown in the proof of $(iii)$, we can prove the item (iv) by item (i) and  \eqref{RIPCondition1}. We here omit  the detail proof.

$(\bm{v})$  The idea of the proof is the argument in \cite[ Step 2]{ge2021new}. 
By \eqref{lem:CrossItem.eq1} and the fact that
$\|\bm{D}_{S^c}^{\top}\bm{h}\|_{\infty}\leq \|\bm{D}_{S}^{\top}\bm{h}\|_{1}/s\leq\|\bm{D}_{S}^{\top}\bm{h}\|_{2}/\sqrt{s}$, we have
\begin{align*}
&\|\bm{D}_{S^c}^{\top}\bm{h}\|_2^2
\leq \|\bm{D}_{S^c}^{\top}\bm{h}\|_{1}\|\bm{D}_{S^c}^{\top}\bm{h}\|_{\infty}\nonumber\\
&\leq\big(a\|\bm{D}_{S}^{\top}\bm{h}\|_1+b\|\bm{D}_{S}^{\top}\bm{h}\|_2+c\|\bm{D}_{T^c}^{\top}\bm{x}\|_1+\eta\|{\bm {Ah}}\|_2+\gamma+\alpha\|\bm{D}_{S^c}^{\top}\bm{h}\|_2\big)
\frac{\|\bm{D}_{S}^{\top}\bm{h}\|_2}{\sqrt{s}}\nonumber\\
&\leq\frac{\alpha\|\bm{D}_{S}^{\top}\bm{h}\|_2}{\sqrt{s}}\|\bm{D}_{S^c}^{\top}\bm{h}\|_2+\frac{a\sqrt{s}+b}{\sqrt{s}}\|\bm{D}_{S}^{\top}\bm{h}\|_2^2+\frac{c\|\bm{D}_{T^c}^{\top
}\bm{x}\|_1+\eta\|{\bm {Ah}}\|_2+\gamma}{\sqrt{s}}\|\bm{D}_{S}^{\top}\bm{h}\|_2 .
\end{align*}
That is,
\begin{align*}
&\bigg(\|\bm{D}_{S^c}^{\top}\bm{h}\|_2-\frac{\alpha\|\bm{D}_{S}^{\top}\bm{h}\|_2}{2\sqrt{s}}\bigg)^2\\
&\leq\bigg(\frac{\alpha^2}{4s}+\frac{a\sqrt{s}+b}{\sqrt{s}}\bigg)\|\bm{D}_{S}^{\top}\bm{h}\|_2^2
+\frac{c\|\bm{D}_{T^c}^{\top}\bm{x}\|_1+\eta\|{\bm {Ah}}\|_2+\gamma}{\sqrt{s}}\|\bm{D}_{S}^{\top}\bm{h}\|_2.
\end{align*}
Then, we obtain
\begin{align}\label{e:anotherupperbound}
&\|\bm{D}_{S^c}^{\top}\bm{h}\|_2\nonumber\\
&\leq \Bigg(\sqrt{\frac{a\sqrt{s}+b}{\sqrt{s}}+\frac{\alpha^2}{4s}}+\frac{\alpha}{2\sqrt{s}}\Bigg)\|\bm{D}_{S}^{\top}\bm{h}\|_2
+\sqrt{\frac{c\|\bm{D}_{T^c}^{\top}\bm{x}\|_1+\eta\|{\bm {Ah}}\|_2+\gamma}{\sqrt{s}}\|\bm{D}_{S}^{\top}\bm{h}\|_2}\nonumber\\
&\leq\Bigg(\sqrt{\frac{a\sqrt{s}+b}{\sqrt{s}}+\frac{\alpha^2}{4s}}+\frac{\alpha+\bar{\varepsilon}}{2\sqrt{s}}\Bigg)\|\bm{D}_{S}^{\top}\bm{h}\|_2
+\frac{1}{2\bar{\varepsilon}}\big(c\|\bm{D}_{T^c}^{\top
}\bm{x}\|_1+\eta\|\bm{Ah}\|_2+\gamma\big),
\end{align}
where the second inequality  comes from the basic inequality $\sqrt{|a||b|}\leq \frac{|a|+|b|}{2}$, and the constant $\bar{\varepsilon}>0$.
\end{proof}

\end{appendices}
\section*{Acknowledgments}
The  project is partially supported by the Natural Science Foundation of China (Nos. 11901037, 72071018), NSFC of Gansu Province, China (Grant No. 21JR7RA511), the NSAF (Grant No. U1830107) and the Science  Challenge Project (TZ2018001).
Authors thanks Professors Xiaobo Qu for  making the pFISTA code available online.


\hskip\parindent

\bibliographystyle{plain}
\bibliography{Unconstrainted_FrameL1_L2_Minimization-2021}

\end{document}